\documentclass[11pt]{article}
\usepackage{graphicx, subfigure}
\usepackage{epsf}
\usepackage{epsfig}
\usepackage{latexsym}
\usepackage{times}
\usepackage{mathptm}
\usepackage{wide}

\newtheorem{lemma}{Lemma}
\newtheorem{theorem}{Theorem}
\newenvironment{proof}{\textit{Proof:}~}{\hfill$\Box$\par\vskip1em}

\begin{document}

\title{Gathering an even number of robots in an odd ring \\without global multiplicity detection \thanks{This work is supported in part by KAKENHI no.22700074.}}
\author{
Sayaka Kamei
\and 
Anissa Lamani
\and 
Fukuhito Ooshita
\and
S\'{e}bastien Tixeuil
}

\date{}

\maketitle

\begin{abstract}
We propose a gathering protocol for an \emph{even} number of robots in a ring-shaped network that allows symmetric but not periodic configurations as initial configurations, yet uses only local weak multiplicity detection. 
Robots are assumed to be anonymous and oblivious, and the execution model is the non-atomic CORDA model with asynchronous fair scheduling. 
In our scheme, the number of robots $k$ must be greater than $8$, the number of nodes $n$ on a network must be odd and greater than $k+3$. 
The running time of our protocol is $O(n^{2})$ asynchronous rounds.
\end{abstract}

\centerline{{\bf Keywords}: Asynchronous Gathering, Local Weak Multiplicity Detection, Robots. }

\section{Introduction}

We consider autonomous robots that are endowed with visibility sensors (but that are otherwise unable to communicate) and motion actuators. Those robots must collaborate to solve a collective task, namely \emph{gathering}, despite being limited with respect to input from the environment, asymmetry, memory, etc. The area where robots have to gather is modeled as a graph and the gathering task requires every robot to reach a single vertex that is unknown beforehand, and to remain there hereafter.

Robots operate in \emph{cycles} that comprise \emph{look}, \emph{compute}, and \emph{move} phases. The look phase consists in taking a snapshot of the other robots positions using its visibility sensors. In the compute phase, a robot computes a target destination among its neighbors, based on the previous observation. The move phase simply consists in moving toward the computed destination using motion actuators. We consider an asynchronous computing model, \emph{i.e.}, there may be a finite but unbounded time between any two phases of a robot's cycle. Asynchrony makes the problem hard since a robot can decide to move according to an old snapshot of the system and different robots may be in different phases of their cycles at the same time.
Moreover, the robots that we consider here have weak capacities: they are \emph{anonymous} (they execute the same protocol and have no mean to distinguish themselves from the others), \emph{oblivious} (they have no memory that is persistent between two cycles), and have no compass whatsoever (they are unable to agree on a common direction or orientation in the ring). 

While most of the literature on coordinated distributed robots considers that those robots are evolving in a \emph{continuous} two-dimensional Euclidean space and use visual sensors with perfect accuracy that permit to locate other robots with infinite precision, a recent trend was to shift from the classical continuous model to the \emph{discrete} model. In the discrete model, space is partitioned into a \emph{finite} number of locations. This setting is conveniently represented by a graph, where nodes represent locations that can be sensed, and where edges represent the possibility for a robot to move from one location to the other. For each location, a robot is able to sense if the location is empty or if robots are positioned on it (instead of sensing the exact position of a robot). Also, a robot is not able to move from a position to another unless there is explicit indication to do so (\emph{i.e.}, the two locations are connected by an edge in the representing graph). 
The discrete model permits to simplify many robot protocols by reasoning on finite structures (\emph{i.e.}, graphs) rather than on infinite ones. 

\paragraph{\textbf{Related Work.}}\label{sec:RW}

In this paper, we focus on the \emph{gathering} problem in the discrete setting where a set of robots has to gather in one single location, not defined in advance, and remain on this location \cite{Izumi10,Flocchin04,Klasing08,Klasing08-j,KameiLOT11,ASN11c}. 
Several deterministic algorithms have been proposed to solve the gathering problem in a ring-shaped network, which enables many problems to appear due to the high number of symmetric configurations. The case of anonymous, asynchronous and oblivious robots was investigated only recently in this context. It should be noted that if the configuration is periodic and edge symmetric, no deterministic solution can exist~\cite{Klasing08-j}. The first two solutions \cite{Klasing08-j,Klasing08} are complementary: \cite{Klasing08-j} is based on breaking the symmetry whereas \cite{Klasing08} takes advantage of symmetries. However, both \cite{Klasing08-j} and \cite{Klasing08} make the assumption that robots are endowed with the ability to distinguish nodes that host one robot from nodes that host two robots or more in the entire network (this property is referred to in the literature as \emph{global} weak multiplicity detection). This ability weakens the gathering problem because it is sufficient
  for a protocol to ensure that a single multiplicity point exists to have all robots gather in this point, so it reduces the gathering problem to the creation of a single multiplicity point. Nevertheless, the case of an even number of robots proved difficult~\cite{HIKIW08c,ASN11c} as more symmetric situations must be taken into account.

Investigating the feasibility of gathering with weaker multiplicity detectors was recently addressed in \cite{Izumi10,KameiLOT11}. In those papers, robots are only able to test that their current hosting node is a multiplicity node (\emph{i.e.} hosts at least two robots). This assumption (referred to in the literature as \emph{local} weak multiplicity detection) is obviously weaker than the global weak multiplicity detection, but is also more realistic as far as sensing devices are concerned. The downside of \cite{Izumi10} compared to \cite{Klasing08} is that only rigid configurations (\emph{i.e.} non symmetric configuration) are allowed as initial configurations (as in  \cite{Klasing08-j}), while \cite{Klasing08} allowed symmetric but not periodic configurations to be used as initial ones. Also, \cite{Izumi10} requires that $k < n/2$ even in the case of non-symmetric configurations, where $k$ denotes the number of robots and $n$ the size of the ring, respectively. By contrast,
  \cite{KameiLOT11} proposed a gathering protocol that could cope with symmetric yet aperiodic configurations and only made use a local weak multiplicity detector, allowing $k$ to be between $3$ and $n-3$. However, \cite{KameiLOT11} requires an odd number of robots, which permits to avoid a number of possibly problematic symmetric configurations.

\paragraph{\textbf{Our Contribution.}}

We propose a gathering protocol for an \emph{even} number of robots in a ring-shaped network that allows symmetric but not periodic configurations as initial configurations, yet uses only local weak multiplicity detection. Robots are assumed to be anonymous and oblivious, and the execution model is the non-atomic CORDA model with asynchronous fair scheduling. For the even number of robots setting, our protocol allows the largest set of initial configurations (with respect to impossibility results) yet uses the weakest multiplicity detector to date. In our scheme, $k$ must be greater than $8$, $n$ must be odd and greater than $k+3$. The running time of our protocol is $O(n^{2})$ asynchronous rounds.

\paragraph{\textbf{Outline.}} The paper is organized as follow: we first define our model in Section \ref{sec:Model}, we then present our algorithm in Section \ref{sec:Algo}. The proofs of correctness are given in Section \ref{sec:proof}. Finally we conclude the paper in Section \ref{sec:conclusion}.

\section{Preliminaries}
\paragraph{\textbf{System Model.}}\label{sec:Model}
We consider here the case of an anonymous, unoriented and undirected ring of $n$ nodes $u_{0}$,$u_{1}$,..., $u_{(n-1)}$ such as $u_{i}$ is connected to both $u_{(i-1)}$ and $u_{(i+1)}$ and $u_{(n-1)}$ is connected to $u_{0}$.
We assume $n$ is odd. 
Note that since no labeling is enabled (anonymous), there is no way to distinguish between nodes, or between edges. 

On this ring, $k$ robots operate in distributed way in order to accomplish a common task that is to gather in one location not known in advance. 
We assume that $k$ is even.
The set of robots considered here are \textit{identical}; they execute the same program using no local parameters and one cannot distinguish them using their appearance, and are \textit{oblivious}, which means that they have no memory of past events, they can't remember the last observations or the last steps taken before. 
In addition, they are unable to communicate directly, however, they have the ability to sense the environment including the position of the other robots. 
Based on the configuration resulting of the sensing, they decide whether to move or to stay idle. 
Each robot executes cycles infinitely many times, 
~(1)~first, it catches a sight of the environment to see the position of the other robots (look phase), 
~(2)~according to the observation, it decides to move or not (compute phase), 
~(3)~if it decides to move, it moves to its neighbor node towards a target destination (move phase). 
At instant $t$, a subset of robots is activated by an entity known as \textit{the scheduler}.
The scheduler can be seen as an external entity that selects some robots for execution, this scheduler is considered to be fair, which means that, all robots must be activated infinitely many times. 
The \textit{CORDA model} \cite{Pre01} enables the interleaving of phases by the scheduler 
(For instance, one robot can perform a look operation while another is moving).
The model considered in our case is the CORDA model with the following constraint: the Move operation is instantaneous \textit{i.e.} when a robot takes a snapshot of its environment, it sees the other robots on nodes and not on edges. However, since the scheduler is allowed to interleave the different operations, robots can move according to an outdated view since during the Compute phase, some robots may have moved.

During the process, some robots move, and at any time occupy nodes of the ring, their positions form a configuration of the system at that time. 
We assume that, at instant $t=0$ (\textit{i.e.}, at the initial configuration), some of the nodes on the ring are occupied by robots, such as, each node contains at most one robot. 
If there is no robot on a node, we call the node \textit{empty node}.
The segment $[u_p,u_q]$ is defined by the sequence ($u_p,u_{p+1},\cdots, u_{q-1},u_q$) of consecutive nodes in the ring, 
such as all the nodes of the sequence are empty except $u_p$ and $u_q$ that contain at least one robot.
The distance $D_p^t$ of segment $[u_p,u_q]$ in the configuration of time $t$ is equal to the number of nodes in $[u_p,u_q]$ minus $1$.
We define a \textit{hole} as the maximal set of consecutive empty nodes. 
That is, in the segment $[u_p,u_q]$, $(u_{p+1},\cdots,u_{q-1})$ is a hole. 
The size of a hole is the number of empty nodes that compose it, the border of the hole are the two empty nodes who are part of this hole, having one robot as a neighbor.

We say that there is a \textit{tower} at some node $u_{i}$, if at this node there is more than one robot (Recall that this tower is distinguishable only locally). 

When a robot takes a snapshot of the current configuration on node $u_{i}$ at time $t$, it has a \textit{view} of the system at this node.
In the configuration $C(t)$, we assume $[u_1,u_2]$,$[u_2,u_3]$,$\cdots$, $[u_w,u_1]$ are consecutive segments in a given direction of the ring.
Then, the view of a robot on node $u_1$ at $C(t)$ is represented by \\$(\max\{(D_1^t,D_2^t,\cdots,D_w^t),(D_w^t,D_{w-1}^t,\cdots,D_1^t)\}, m_1^t)$, where
$m_1^t$ is true if there is a tower at this node, and sequence $(a_i,a_{i+1},\cdots, a_j)$ is larger than $(b_i,b_{i+1},\cdots,b_j)$ if there is $h (i\leq h\leq j)$ such that $a_l=b_l$ for $i\leq l \leq h-1$ and $a_h>b_h$.
It is stressed from the definition that robots don't make difference between a node containing one robot and those containing more. 
However, they can detect $m^t$ of the current node, \textit{i.e.} whether they are alone on the node or not (they have a local weak multiplicity detection). 

When $(D_1^t,D_2^t,\cdots,D_w^t)=(D_w^t,D_{w-1}^t,\cdots,D_1^t)$, we say that the view on $u_i$ is \textit{symmetric}, otherwise we say that the view on $u_{i}$ is \textit{asymmetric}. 
Note that when the view is symmetric, both edges incident to $u_{i}$ look identical to the robot located at that node. 
In the case the robot on this node is activated we assume the worst scenario allowing the scheduler to take the decision on the direction to be taken.  

Configurations that have no tower are classified into three classes in \cite{Klasing08-j}.
A configuration is said to be \textit{periodic} if it is represented by a configuration of at least two copies of a sub-sequence.
A configuration is said to be \textit{symmetric} if the ring contains a single axis of symmetry. 
Otherwise, the configuration is said to be \textit{rigid}. 
For these configurations, the following lemma is proved in \cite{Klasing08-j}.
\begin{lemma}
If a configuration is rigid, all robots have distinct views. If a configuration is symmetric and non-periodic, there exists exactly one axis of symmetry.
\end{lemma}

This implies that, if a configuration is symmetric and non-periodic, at most two robots have the same view.

We now define some useful terms that will be used to describe our algorithm. 
We denote by the \textit{inter-distance} $d$ the minimum distance taken among distances between each pair of distinct robots (in term of the number of edges). 
Given a configuration of inter-distance $d$, a \textit{$d$.block} is any maximal elementary path where there is a robot every $d$ edges. 
The border of a $d$.block consists in the two external robots of the $d$.block. 
The size of a $d$.block is the number of robots that it contains. 
We call the $d$.block whose size is biggest the \textit{biggest $d$.block}.
A robot that is not in any $d$.block is said to be an \textit{isolated robot}. 

We evaluate the time complexity of algorithms with asynchronous rounds. An asynchronous round is defined as the shortest fragment of an execution where each robot performs a move phase at least once.

\paragraph{\textbf{Problem to be solved.}} The problem considered in our work is the gathering problem, where $k$ robots have to gather in one location not known in advance before stopping there forever.

\section{Proposed Algorithm}\label{sec:Algo}
We propose in this section an algorithm that solves the gathering problem starting from any non-periodic configuration on a ring provided that $n$ is odd, $k$ is even, $k>8$ and $n>k+3$.

The algorithm comprises three main phases as follow:
\begin{itemize}
\item \textbf{Phase $1$}. The aim of this phase is either ($i$) to reach one of the special configurations defined in the set $C_{sp}$ (refer to Sub-Section \ref{Subsec: ph2}) or ($ii$) to reach a symmetric configuration where there are two 1.blocks of size $k/2$ at distance $2$. The initial configuration of this phase is any non-periodic configuration without towers.
\item \textbf{Phase $2$}. The aim of this phase is to reach a symmetric configuration that contains exactly two 1.blocks of size $k/2$ at distance $2$ from each other. The initial configuration of this phase is in $C_{sp}$.
\item \textbf{Phase $3$}. During this phase, robots perform the gathering such that at the end all robots are on the same node.  
\end{itemize}

When the configuration is symmetric, and since the number of nodes is odd, we are sure that the axes of symmetry passes through one node $S$ and one edge. 
Additionally, because the number of robots is even, there is no robot on $S$. The size of the hole including $S$ is odd. Let us call such a hole the {\em Leader hole} and let us refer to it by $H$. The other hole on the axes of symmetry is called {\em Slave hole}. Note that, theses hole can be on inside a $d$.block.



\subsection{Phase $1$: Algorithm to build a single $1$.block or two $1$.blocks of the same size}

Starting from a non periodic configuration without tower, the aim of this phase is to reach either one of the special configurations defined in $C_{sp}$ (refer to Sub-Section \ref{Subsec: ph2}). 
The idea of this phase is the following: In the initial configuration, in the case all the $d$.blocks have the same size, robots move such that there will be at least two $d$.blocks in the configuration that have different size. 
Robots then move towards the closest biggest $d$.block. 
In order to avoid creating periodic configurations, only some robots are allowed to move. 
These robots are the ones that have the biggest view. 
Depending on the nature of the configuration (symmetric or not symmetric), one $d$.block (respectively two $d$.blocks) becomes the biggest $d$.block in the configuration (let refer to the set of these $d$.block by target blocks). 
These $d$.blocks are then the target of all the other robots that move in order to join them. 
When all the robots are in a $d$.block part of target blocks then if $d>1$, some robots move in order to decrease $d$. We distinguish the following configurations:

\begin{itemize}
\item {\em BlockDistance} Configuration. 
In this configuration, there are only two $d$.blocks of the same size ($k/2$) or a single $d$.block of size $k$ 
such that in both cases $d>1$. Note that the configuration is symmetric and does not contain any isolated robot.
The robots allowed to move are the ones that are neighbors of $H$.
Their destination is their adjacent empty node towards the opposite direction of $H$.

\item {\em BlockMirror} Configuration. 
In such a configuration there are only $d$.blocks of the same size and no isolated robots. The number of $d$.blocks is bigger than $2$.
Two sub configurations are possible as follow:
\begin{itemize}
\item {\em BlockMirror1} Configuration.
The configuration is in this case not symmetric.
The robot that is allowed to move is the one that is part of the set of robots that are the closest to a $d$.block (let refer to this set by $S$), having the biggest view among the robots in $S$. Its destination is its adjacent empty node towards the neighboring $d$.block. 
\item {\em BlockMirror2} Configuration.
The configuration is in this case symmetric.
Let the $d$.blocks that are neighbors to $H$ or including $H$ be the guide blocks.
The robots allowed to move are the ones that share a hole other than $H$ with the guide blocks.
Their destinations are their adjacent empty node towards the closest guide block.
\end{itemize}

\item Configuration of type {\em BigBlock}.
In this configuration, the configuration is neither {\em BlockMirror} nor {\em BlockDistance}. 
Then, there is at least one $d$.block that has the biggest size. 
Two sub configurations are defined as follow:
\begin{itemize}
\item Configuration of type {\em BigBlock1}. 
In this case there is at least one isolated robot that shares a hole with one of the biggest $d$.blocks. Two sub-cases are possible as follow:
      \begin{itemize}
             \item Configuration of type {\em BigBlock1-1}. The configuration is not symmetric and contains either ($i$) two $1$.blocks of the same size $(k-2)/2$ and two isolated robots that share a hole together or ($ii$) one $1$.block of size ($k-2$) and two isolated robots that share a hole together. The robot that is allowed to move in both cases is the one that is the farthest to the neighboring $1$.block. Its destination is its adjacent empty node towards the neighboring $1$.block.
             \item Configuration of type  {\em BigBlock1-2}. This configuration is different from  {\em BigBlock1-1} and includes all the other configurations of {\em BigBlock1}. The robots allowed to move are part of the set of robots that share a hole with a biggest d.block such that they are the closest ones to a biggest $d$.block. 
If there exists more than one such robot, then only robots with the biggest view among them are allowed to move.
Their destination is their adjacent empty node towards one of the nearest neighboring biggest $d$.blocks.
      \end{itemize}

\item Configuration of type {\em BigBlock2}. 
In this case there is no isolated robot that is neighbor to a biggest $d$.block. 
The robots allowed to move are the ones that are the closest to a biggest $d$.block. 
If there exist more than one robot allowed to move, then only robots with the biggest view among them can move. 
Their destination is their adjacent empty node towards one of the nearest neighboring biggest $d$.blocks.
\end{itemize}

\end{itemize}    


\subsection{Phase $2$: Algorithm to build a configuration that contains two $1$.blocks of the same size at distance $2$.}\label{Subsec: ph2}

When the configuration is symmetric, the two $1$.blocks that are neighbors of the Leader hole ($H$) (respectively Slave Hole) are called {\it the Leader block} (respectively {\it the Slave block}).
To simplify the explanation we assume in the following that an isolated robot is also a single $1$.block of size $1$.

The configuration set $C_{sp}$ is partitioned into nine subsets:
{\em Start}, {\em Even-T}, {\em Split-S}, {\em Split-A}, {\em Odd-T}, {\em Block}, {\em Biblock}, {\em TriBlock-S} and {\em TriBlock-A}.
In the following, we provide their definition and the behavior of robots in each configuration below.

\begin{figure}[t]
 \begin{minipage}{.46\linewidth}
  \centering\epsfig{figure=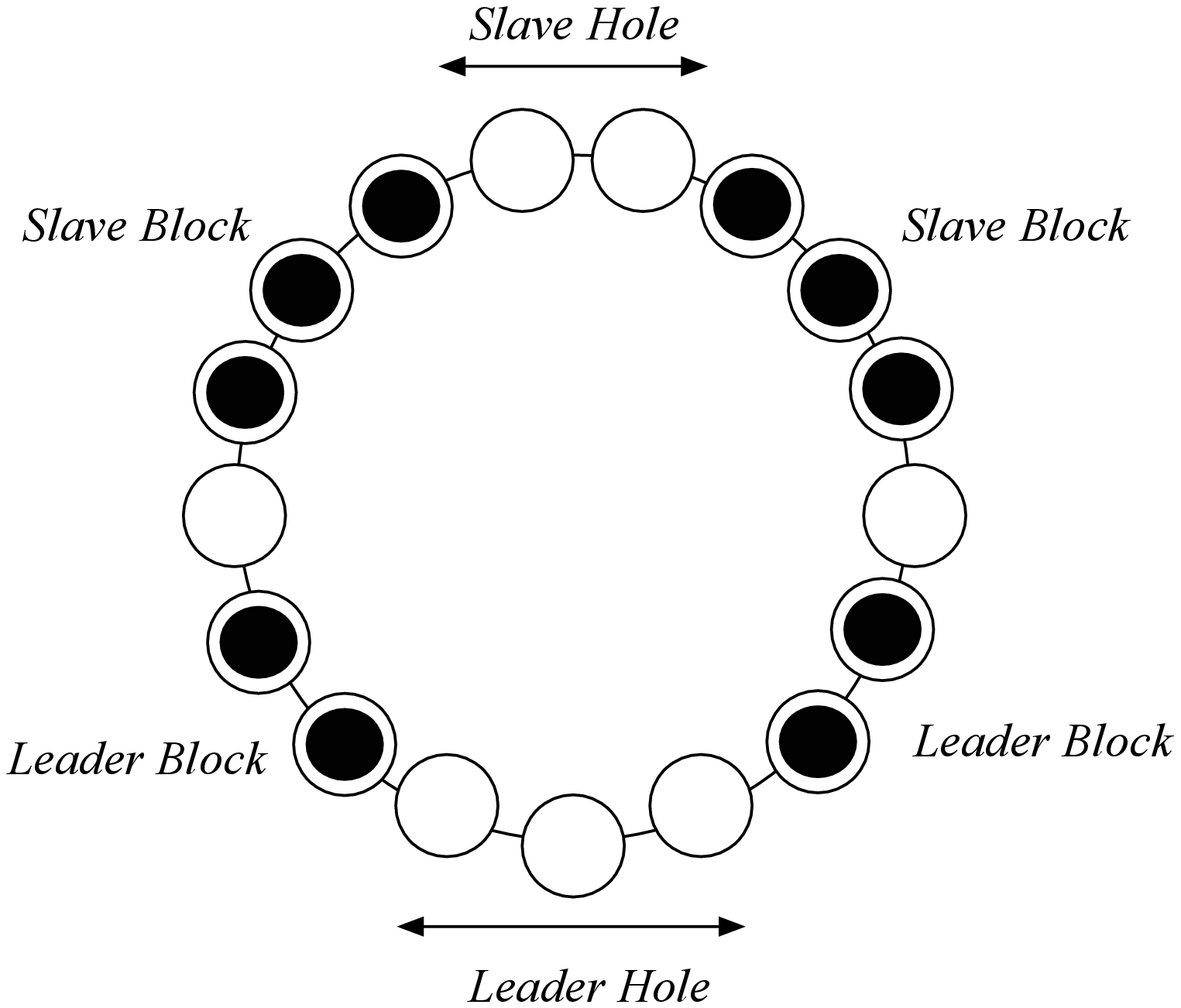,width=3.8cm}
  \caption{$Split-S$ Configuration\label{h-block}}
 \end{minipage} \hfill
\begin{minipage}{.46\linewidth}
 \centering\epsfig{figure=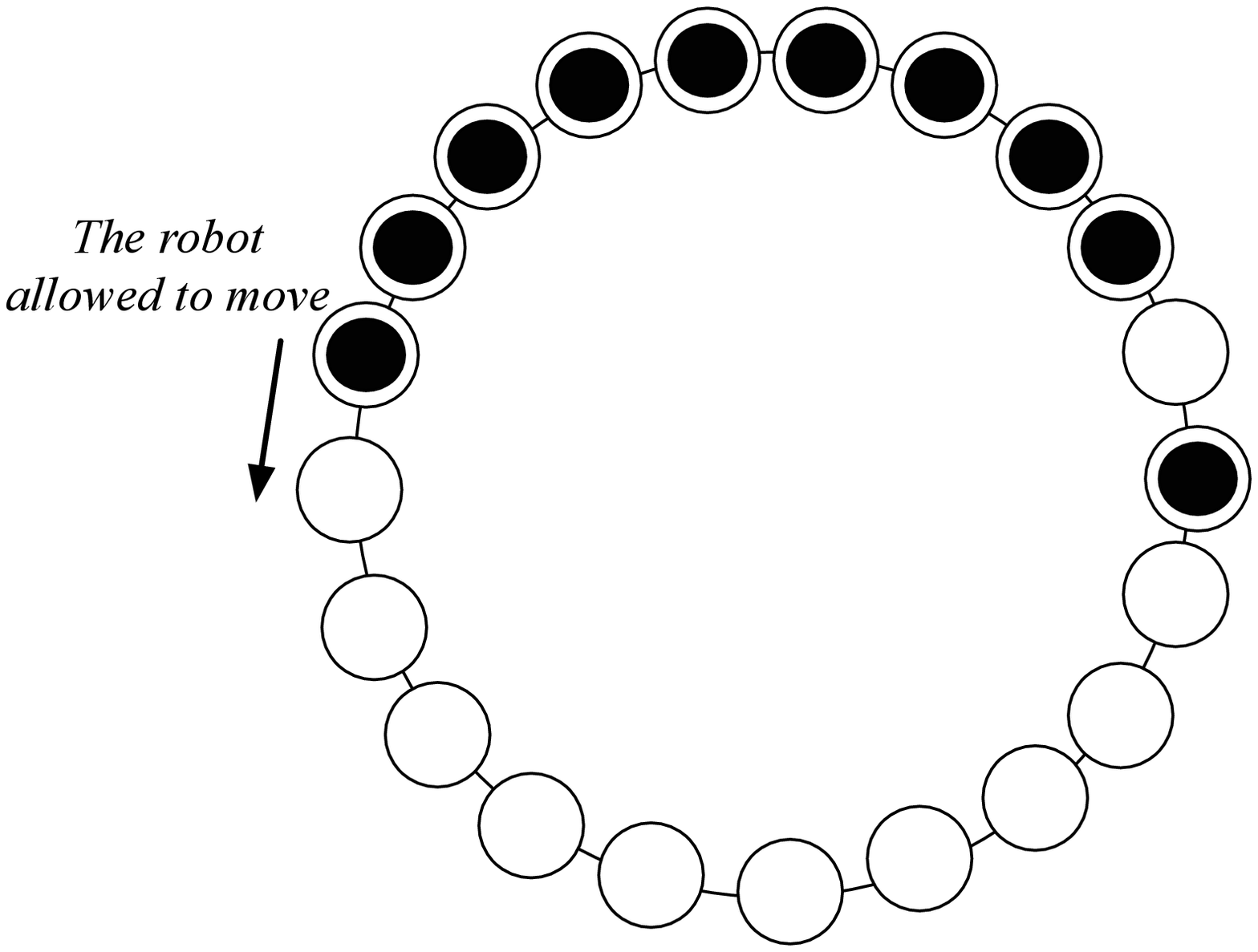,width=3.5cm}
  \caption{$Biblock$ Configuration\label{BIB}}
 \end{minipage} \hfill
 \begin{minipage}{.46\linewidth}
  \centering\epsfig{figure=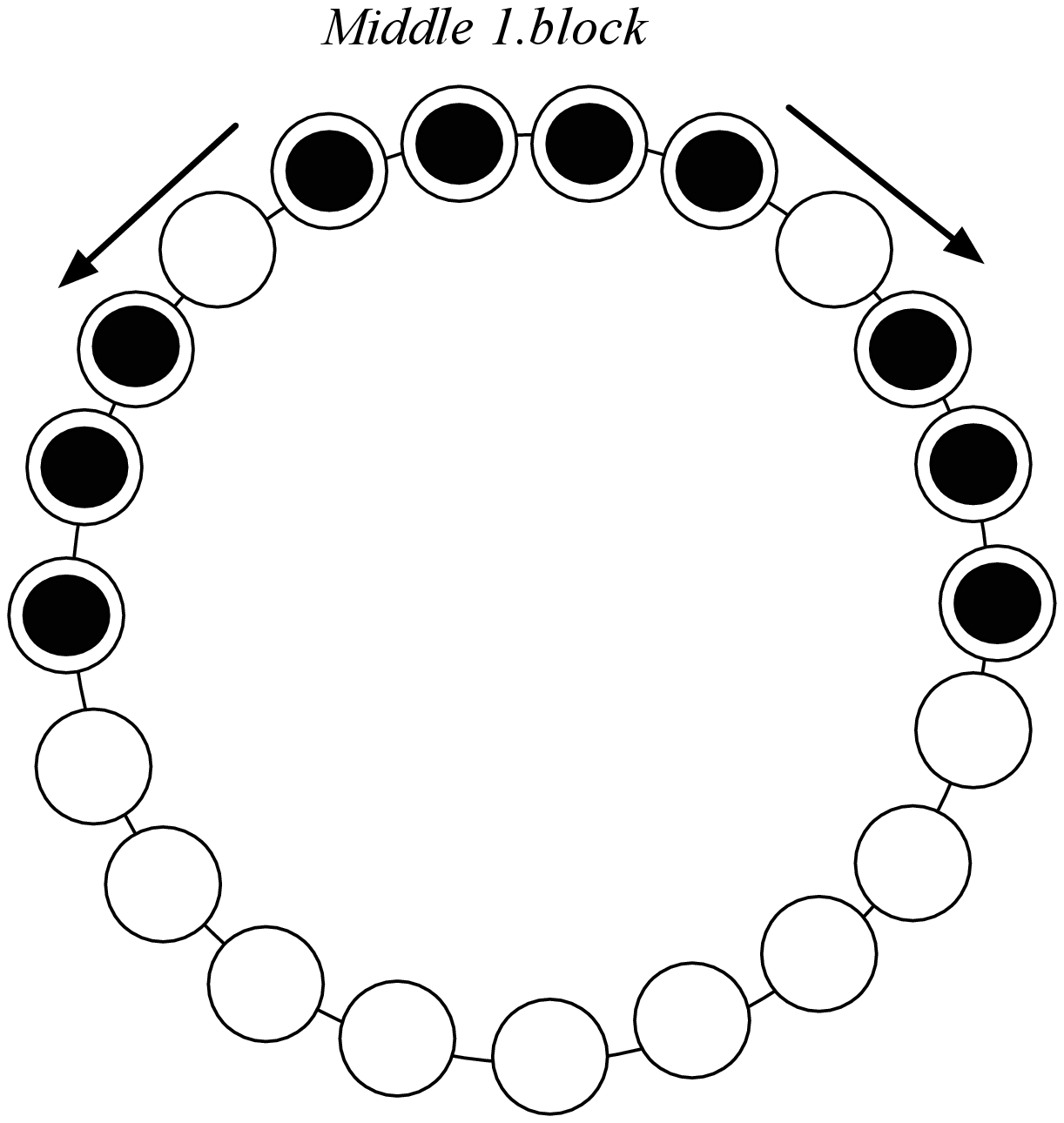,width=3cm}
  \caption{$TriBlock-S$ Configuration\label{TRIBS}}
 \end{minipage}\hfill
 \begin{minipage}{.46\linewidth}
  \centering\epsfig{figure=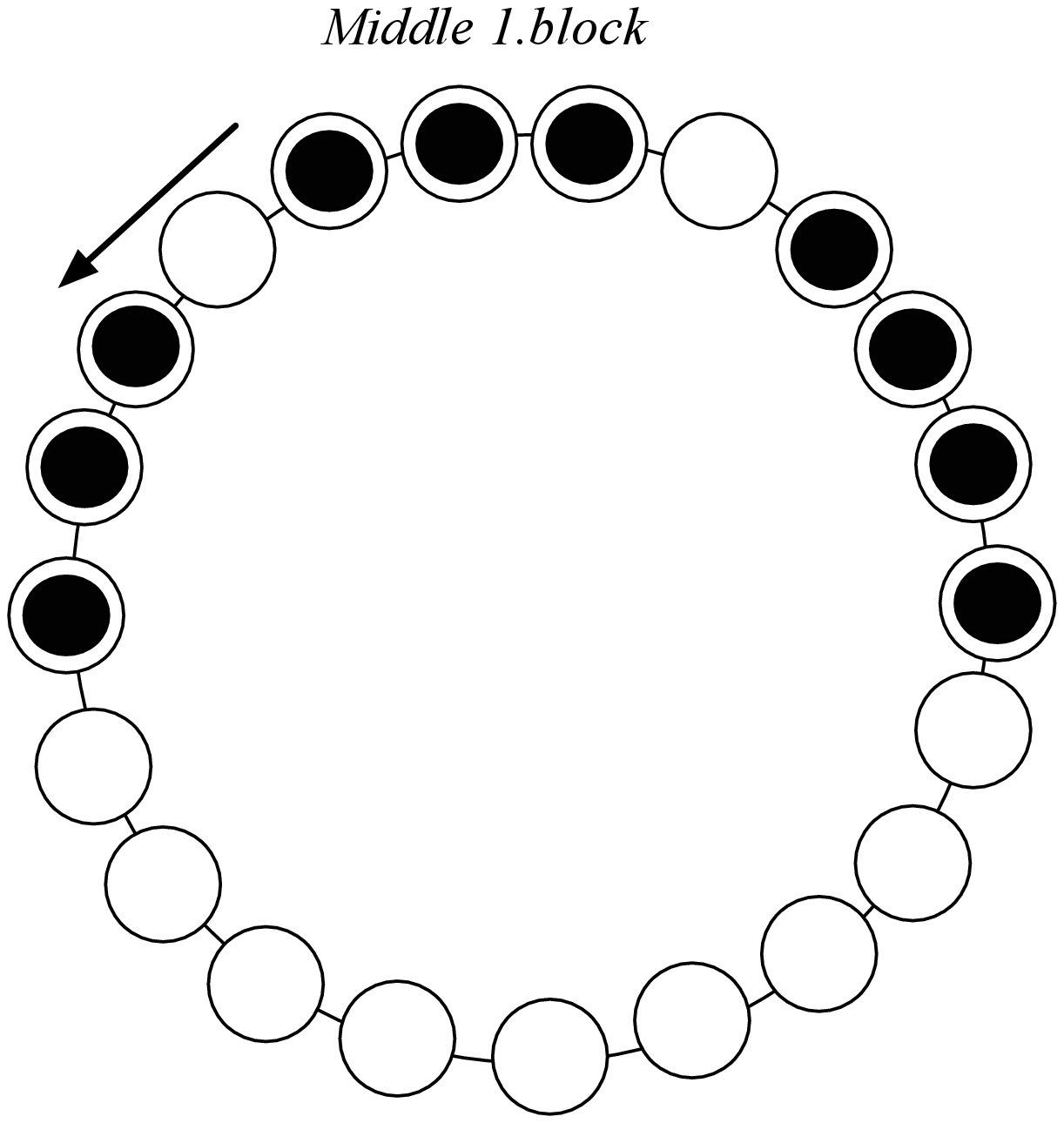,width=3cm}
  \caption{$TriBlock-A$ Configuration\label{TRIBA}}
 \end{minipage}\hfill
  \begin{minipage}{.46\linewidth}
  \centering\epsfig{figure=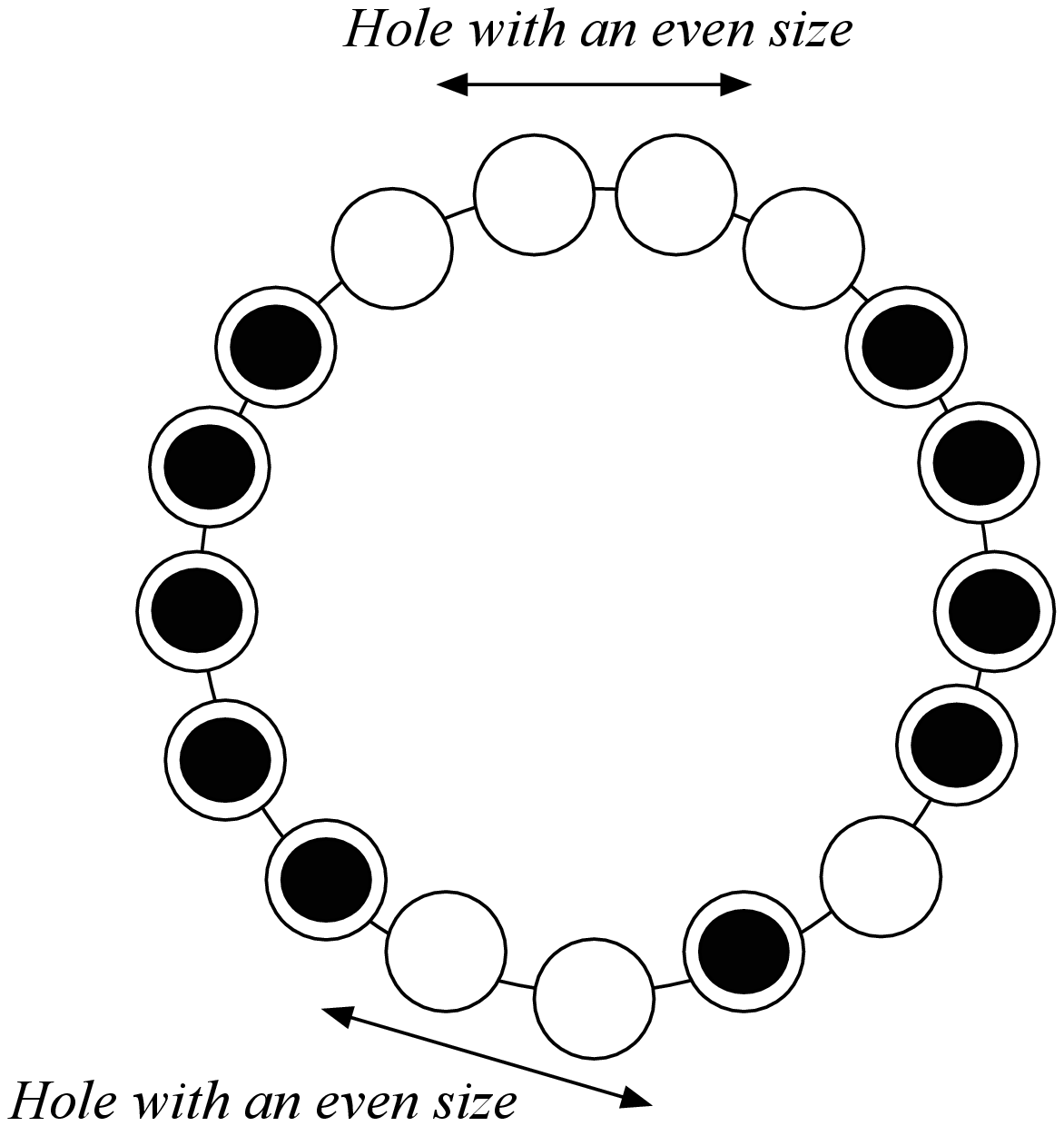,width=3.4cm}
  \caption{$Even-T$ Configuration\label{EvenT}}
 \end{minipage}\hfill
  \begin{minipage}{.46\linewidth}
  \centering\epsfig{figure=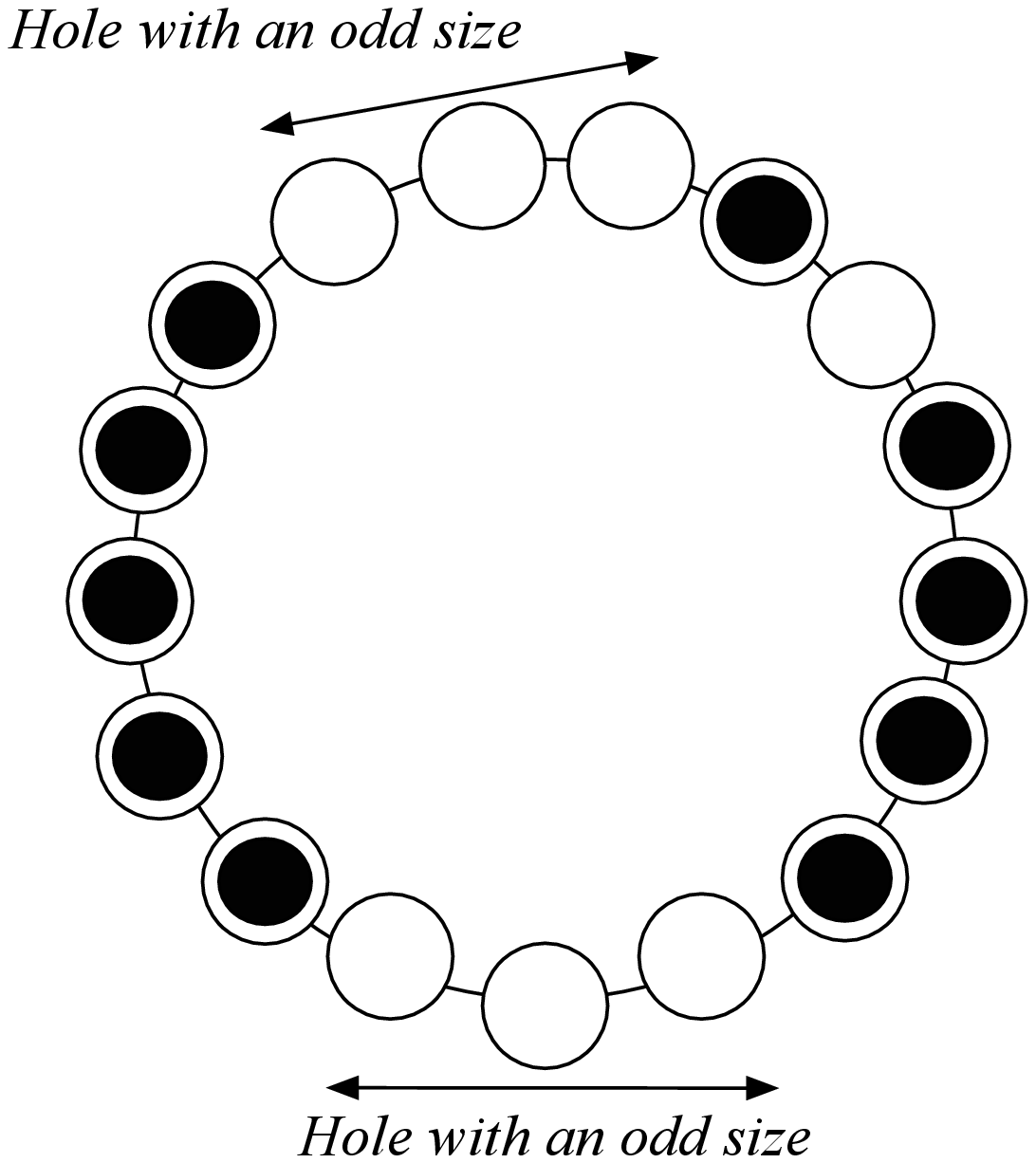,width=3.2cm}
  \caption{$Odd-T$ Configuration\label{OddT}}
 \end{minipage}\hfill
\end{figure}

\begin{enumerate}
\item {\em Start} Configuration.
This configuration is symmetric and contains two $1$.blocks with size $k/2$ but not being at distance $2$. 
The robots allowed to move are the two robots that are at the border of the $1$.blocks neighboring to the Leader hole. 
Their destination is their adjacent empty node in the opposite direction of the $1$.block they belong to.

\item \label{CC} {\em Even-T} Configuration.
The configuration is in this case not symmetric and contains three $1$.blocks of size respectively $k/2$, $(k/2)-1$ and $1$.
Additionally, the $1$.block of size $1$ is at distance $2$ from the $1$.block of size $k/2-1$. The number of holes between the $1$.block of size $1$ and the one of size $k/2$ is even. Note that this is also the case for the hole between the $1$.block of size $(k/2)-1$ and the one of size $k/2$.
The only robot allowed to move is the one that is at the border of the $1$.block of size $k/2$ sharing a hole with the $1$.block of size $1$. 
Its destination is its adjacent empty node in the opposite direction of the $1$.block it belongs to.

\item {\em Split-S} configuration.
The configuration is symmetric and contains four $1$.blocks such that the $1$.blocks on the same side of the axes of symmetry are at distance $2$ (refer to Figure \ref{h-block}). 
The robots allowed to move in this case are the ones that are at the border of the Slave block sharing a hole with the Leader block. 
Their destination is their adjacent empty node towards the Leader block.

\item {\em Split-A} Configuration. 
The configuration is not symmetric and contains four $1$.blocks and exactly one hole of an even size.
Let the $1$.blocks that are neighbors of the hole of an even size be $S1$ and $S2$, and let the other two $1$.blocks be $L1$ and $L2$. $S1$ and $L1$ (respectively $S2$ and $L2$) are at distance $2$ from each other. 
Then, $|S1|=|S2|+1$, $|L2|=|L1|+1$, $|S1|+|L1|=|S2|+|L2|=k/2$. The size of the hole between $L1$ and $L2$ is odd, 
and the distances between $L1$ and $S1$ (respectively between $L2$ and $S2$) is equal to $2$.
In this case, only one robot is allowed to move. This robot is the one at the border of $S1$ sharing a hole with $L1$. 
Its destination its adjacent empty node towards $L1$.

\item {\em Odd-T} Configuration. 
The configuration is not symmetric and contains three $1$.blocks of size respectively $k/2$, $(k/2)-1$ and $1$.
Additionally, the $1$.block of size $1$ is at distance $2$ from the $1$.block of size $k/2-1$. Observe that this configuration is different from {\em Even-T} Configuration since all the holes in the ring have an odd size. 
The only robot allowed to move is the one that is part of the $1$.block of size $1$. Its destination is its adjacent empty node towards the $1$.block of size $(k/2)-1$.

\item {\em Block} Configuration. The configuration contains in this case, a single $1$.block of size $k$. Note that the configuration is symmetric. The robots allowed to move are the ones that are at the border of the $1$.block. Their destination is their adjacent node in the opposite direction of the $1$.block they belong to.

\item {\em Biblock} Configuration. This configuration is not symmetric and contains two $1$.blocks $B_1$ and $B_2$ at distance $2$ from each other such that $|B_1|=k-1$ and $|B_2|=1$ (refer to Figure \ref{BIB}). The robot allowed to move is the one that is at the border of the biggest $1$.block not having a neighboring occupied node at distance $2$. Its destination is its adjacent node in the opposite direction of the $1$.block it belongs to.

\item {\em TriBlock-S} Configuration. This configuration is symmetric and contains three $1$.blocks separated by one empty node (refer to Figure \ref{TRIBS}). The robots allowed to move are the ones that are at the border of the $1$.block on the axes of symmetry. Their destination is their adjacent empty node in the opposite direction of the $1$.block they belong to.

\item {\em TriBlock-A} Configuration. This configuration is not symmetric and contains three $1$.blocks ($B_{1}$, $B_2$ and $B_3$) such that there is one $1$.block that is at distance $2$ from both other $1$.blocks (let $B_1$ be this $1$.block, refer to Figure \ref{TRIBA}). $|B_2|=|B_3|+1$. The robot allowed to move is the one that is at the border of $B_1$ and the closest to $B_3$. Its destination is its adjacent empty node in the opposite direction of the $1$.block it belongs to.    
\end{enumerate}

We call the configuration that contains two 1.blocks of the size $k/2$ at distance $2$ {\em Terminal} (refer to Figure \ref{Terminus}).

\subsection{Phase $3$. Algorithm to achieve the gathering} 

During this phase, robots perform the gathering such that at the end all robots are on the same location. The starting configuration of this phase is the {\em Terminal} configuration (refer to Figure \ref{Terminus}).

\begin{figure}[t]
\begin{center}
\epsfig{figure=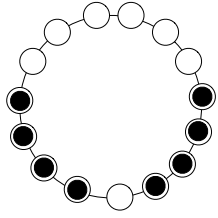,width=3cm}
\caption{Terminal Configuration}\label{Terminus}
\end{center}
\end{figure}

When the {\em Terminal} configuration is reached at the end of the second phase (or when the configuration is built in the first phase). The only robots that can move are the ones that are at the extremity of a $1$.block being neighbors of a hole of size $1$. Since the configuration is symmetric there are exactly two robots allowed to move. Two cases are possible as follow:
\begin{enumerate}
\item The scheduler activates both robots at the same time. In this case the configuration remains symmetric and a tower is created on the axes of symmetry.
\item The scheduler activates only one robot. In this case the configuration that is reached becomes asymmetric and contains two 1.blocks $B_1$ and $B_2$ at distance $2$ such that $|B_2|=|B_1|-2$ (one robot from $B_2$ has moved and joined $B_1$, let this robot be $r_1$). Note that this configuration is easily recognizable by robots. The robot that is in $B_1$ being neighbor of $r_1$ is the one allowed to move. Its destination is its adjacent occupy node towards $r_1$. Note that once it moves, the configuration becomes symmetric and a tower is created on the node that is on the axes of symmetry. 
Let us refer to such symmetric configuration with a tower {\em Target} configuration.
\end{enumerate}

In the {\em Target} configuration, robots that are part of the tower are not allowed to move anymore. For the other robots, they can only see an odd number of robots in the configuration since they are enable to see the tower on the axes of symmetry (recall that they have a local week multiplicity detection). In addition, since they are oblivious, they cannot remember their number before reaching {\em Target} configuration. On the other hand, an algorithm has been proposed in \cite{KameiLOT11} that solves the gathering problem from such a configuration (Phase $2$ in \cite{KameiLOT11}). Robots can then execute the same algorithm to perform the gathering. 

\section{Proof of Correctness}\label{sec:proof}
In the following, we prove the correctness of our algorithm. We define an \textit{outdated robot} as the robot that takes a snapshot at instant $t$ but moves only at instant $t+j$ where $j>0$. Thus, when such a robot moves, it does so based on an outdated perception of the configuration.
Additionally, we define an \textit{outdated robot with an incorrect target} as the outdated robot that its target destination is incorrect in the current configuration, i.e., if the robot takes a new snapshot to the current configuration, the computed destination is different.

 \subsection{Phase $1$}
 
We prove in this sub-section the correctness of Phase $1$ Algorithm.
 
 \begin{lemma}\label{NoSymBigBlock}
Starting from a symmetric {\em BigBlock} configuration without any outdated robots, if exactly one robot moves then the configuration reached is not symmetric.
\end{lemma}

\begin{proof}
Let refer to the symmetric {\em BigBlock} configuration by $C$. Since the configuration is symmetric, there are two robots $A$ and $B$ such that their destination is its adjacent empty node towards the biggest $d$.block. Two cases are possible as follow:
\begin{itemize}
\item The biggest target $d$.blocks of robots $A$ and $B$ is the same (let refer to this block by $D$). 
In this case, $D$ is on the axes of symmetry. 
If the scheduler activates only one robot (let this robot be the robot $B$), then $B$ either joins $D$ or it becomes an isolated robot once it moves.
If $B$ becomes an isolated robot, then it is the closest robot to a biggest d.block (otherwise there were another robot $R$ which was the closest in the previous configuration). Thus, the configuration cannot be in this case symmetric.
When $B$ joins $D$, then there is only one biggest $d$.block in the configuration that is $D$.
Observe that $A$ is the only robot that is the closest to $D$. 
Thus, in this case too, the configuration reached is not symmetric.
\item The biggest target $d$.blocks of robots $A$ and $B$ are different. 
Suppose that, they are respectively $D_1$ and $D_2$. 
When the scheduler activates only one robot ($B$), then it is clear that, if only $B$ moves, the configuration reached is not symmetric before $B$ joins $D_2$ (since $B$ is the only one that is the closest to a biggest d.block). 

Let consider now the case where $B$ joins $D_2$. 
Once $B$ is part of $D_2$, if the configuration that is reached is symmetric, then we are sure that $D_2$ is on the axes of symmetry since it is the only biggest $d$.block in the configuration. 
Two sub-cases are possible:
\begin{itemize} 
\item When $B$ joins $D_2$, the size of $D_2$ becomes odd. 
Since $D_2$ is on the axes of symmetry, there is one occupied node on the axes of symmetry. 
That is impossible since the number of robots in the system is even. 
Thus, we are sure that the configuration that is reached is not symmetric.
\item When $B$ joins $D_2$, the size of $D_2$ becomes even. 
Note that, the size of $D_1$ is odd. 
Since $C$ was symmetric, we are sure that there was an even number of biggest $d$.blocks of size $|D_1|$ in $C$.
Additionally, there were no biggest $d$.blocks on the axes of symmetry (if there was a biggest $d$.block on the axes of symmetry in $C$, then there were two occupied nodes on the axes of symmetry since the size of such d.blocks is odd. That is impossible since the number of nodes in the ring is odd).
Thus, just after $B$ joins $D_2$, there is exactly one biggest $d$.block $D_2$ and odd number of $d$.blocks having the same size as $D_1$. 
If the configuration that is reached is symmetric, there will be one $d$.block having the same size as $D_1$ that is on the axes of symmetry (since their number is odd). 
On the other hand, the size of such a $d$.block is odd. 
That means that there is one occupied node that is on the axes of symmetry.
This is impossible since the number of robots in the ring is even. 
Hence, we can deduce that the configuration that is reached in this case too is not symmetric
\end{itemize}      
\end{itemize}

From the cases above, we can deduce that, if exactly one robot moves from a symmetric {\em BigBlock} configuration, then the configuration reached is not symmetric and the lemma holds.
\end{proof}

\begin{lemma}\label{outdated}
By the behavior of Phase $1$, starting from any non-periodic initial configuration without towers, if the configuration does not become {\em BlockMirror2}, then there is at most one outdated robot with an incorrect target.
\end{lemma}
\begin{proof}
Assume that a robot $A$ becomes an outdated robot.
Then, two robots $A$ and $B$ were allowed to move in a symmetric configuration $C$ such that the scheduler activates both robots at the same time, however only $B$ moves.
We prove that if the configuration becomes symmetric before $A$ moves, the new outdated robot other than $A$ cannot have incorrect target, or the configuration does not become symmetric before $A$ moves.

\begin{itemize}
\item $C$ is a configuration of type {\em BlockDistance}.
Note that when the configuration is of type {\em BlockDistance}, it contains either a single d.block of size $k$ or two d.blocks of size $k/2$.  In both cases $H>d$ otherwise the configuration is either periodic or does not contain two d.blocks of the same size.
Once $B$ moves, the value of $d$ decreases. Let the new value of the interdistance be $d'$. Note that $d'=d-1$.
Then, there exist exactly one $d'$.block in the configuration. Let refer to this $d'$.block by $D$. All the other robots become isolated robots.
By the behavior of Phase $1$, the closest robot to $d'$.block moves towards it (refer to {\em BigBlock1} configuration). Note that there is only one such robot (the other robots are not allowed to move) and this robot is the one that was part of the same d.block as $B$ at distance $d$ (recall that $H>d$). Let refer to this robot by $E$.
If $A$ decides to move after $E$ executes the {\em Look} phase, then the configuration becomes symmetric and the target of $E$ is correct since it is one of the robots that are allowed to move and its destination remains $D$.
If $A$ does not move then, $D$ is the only biggest $d'$.block in the configuration, when $E$ moves, the size of $D$ increases. $D$ becomes the target $d'$.block of all the robots. When $A$ decides to move, even if a new $d'$.block is created and the configuration becomes symmetric, the target of the robots that took a snapshot before the move of $A$ is correct since the target $d'$.block does not change (there is only one biggest $d'$.block in the configuration).
Therefore, before $A$ moves, there is no other robots that has an incorrect target and the lemma holds.

\item $C$ is a {\em BigBlock} configuration. The two following cases are possible:
\begin{itemize}
\item The target $d$.block of $A$ and $B$ is the same $d$.block $BL$. 
Once $B$ moves, we are sure from Lemma \ref{NoSymBigBlock} that the configuration that is reached is not symmetric. 
If $B$ becomes (or remains) an isolated robot, then if the configuration contains exactly two d.blocks of the same size and two isolated robots that share a hole between each other ({\em BigBlock1-1} configuration), then $A$ is the only robot that can move and in this case the configuration does not contain any outdated robots. For the other cases, $B$ is the only one allowed to move since it is the closest robot to a biggest $d$.block ($BL$). 
When $B$ joins $BL$, the only robot allowed to move is $A$ since $A$ is the only robot that is the closest to $BL$ when $A$ is activated (and even if it is an outdated robot with an incorrect target), $A$ moves towards $BL$. 
Once it moves the configuration that is reached is not symmetric (since $A$ is the only closest robot) and there is no robot in the configuration that has an outdated robot with an incorrect target.
\item The target $d$.blocks of $A$ and $B$ are respectively $D_1$ and $D_2$. When $B$ moves, it is clear that if the configuration is of type {\em BigBlock1-1}, there will be no outdated robots. For the other cases,
according to Lemma \ref{NoSymBigBlock}, if $B$ is the only one to move then when $B$ joins $D_1$, the configuration that is reached is not symmetric. 
Note that $D_1$ becomes the only biggest $d$.block in the configuration. 
We can reach a symmetric case in the following cases:
\begin{itemize}
\item $A$ (the outdated robot with an incorrect target) moves: $A$ joins $D_2$ too (it's like the scheduler activates both $A$ and $B$ at the same time). 
Note that, if there is another robot that took a snapshot before $A$ moves, then its target is correct since even when in the configuration reached when $A$ moves, this robot is allowed to move and it destination remains $D_1$. 
\item There is another robot $R$ that has joined $D_1$ (Note that in this case the size of $D_1$ at the beginning in the {\em BigBlock} configuration was even. 
Thus, when $B$ joined $D_1$, the size of $D_1$ became odd and when $R$ joined it too, it became even). 
$D_1$ becomes the biggest $d$.block in the configuration. 
Note that, in the configuration, there can be only one robot with an incorrect target (Robot $A$ that took a snapshot at the beginning in the {\em BigBlock} configuration, its destination is $B_2$), all the robots that took a snapshot before $A$ moves, even if their view is outdated, their target are correct since the target $d$.block remains $D_1$ (and this even if $A$ moves an joins $D_2$ since even after $A$ move, $D_1$ is bigger than $D_2$).
\end{itemize}
\end{itemize}
\end{itemize}    
From the cases before, we can deduce that there is at most one outdated robot with an incorrect target.
\end{proof}

\begin{lemma}\label{outdated-special}
By the behavior of Phase $1$, starting from any non-periodic {\em BlockMirror2} configuration without any tower, there exist at most one outdated robot with an incorrect target.
\end{lemma}

\begin{proof}
Assume that a robot $A$ becomes an outdated robot.
Then, two robots $A$ and $B$ were allowed to move in a configuration $C$ of type {\em BlockMirror2} such that the scheduler activated both $A$ and $B$, however only $B$ moves.

In $C$, all robots belong to $d$.blocks and all the $d$.blocks have same size $s$.
The destinations of robots $A$ and $B$ are the guide blocks (that are neighbors to the guide hole $H$).
When $B$ moves, it either joins a guide $d$.block $D1$ or it becomes an isolated robot.\\

($i$) In the case $|s|=2$, then there is another robot $B'$ that was in the same $d$.block as $B$ in $C$ that becomes an isolated robot. Two cases are possible as follow:
\begin{itemize}
\item $B$ joins $D1$. $D1$ is the only biggest $d$.block in the configuration ($|D1|=3$), and $B'$ is the only closest isolated robot to the $d$.block of size $3$ ($D1$).
Therefore, before $B'$ joins $D1$, only $B'$ is allowed to move. 
\begin{itemize}
\item If $A$ decides to move before $B'$ joining $D1$ then $A$ joins also a d.block $D2$ and either the configuration becomes symmetric and in this case the target of $B'$ remains $D1$ or the configuration is not symmetric and $B'$ is the only robot allowed to move, its target destination is $D1$. 
\item Let consider the case where $A$ does not move before $B'$ joining $D1$. After $B'$ joins $D1$, $D1$ remains the only biggest $d$.block in the configuration. Note that the target $d$.block of all the robots in the configuration when they are allowed to move (except $A$ which has an outdated view with an incorrect target) is $D1$. If $A$ decides to move, it joins the $d$.block $D2$. Note that $|D1|>|D2|$, thus all the robots that took a snapshot before $A$ moves have a correct target even if their view is outdated.   
\end{itemize}

\item $B$ becomes an isolated robot, then $B$ and $B'$ are the only isolated robots in $C$. In this case $B$ is the only one that is the closest to a $d$.block (otherwise, in the {\em BlockMirror2} configuration, the distance between $B'$ and a $d$.block is smaller than between $B$ and a $d$.block.) and thus, it is the only one allowed to move. When $B$ joins the closest $d$.block $D1$, there will be in the configuration exactly one biggest $d$.block which is $D1$. The target of all the robots (except $A$ which already has an outdated view) when they are allowed to move becomes $D1$). Thus, when $A$ decides to move, all the robots that took a snapshot have a correct target even if it is computed according to an outdated view. 
\end{itemize}

($ii$) In the case $|s|>2$. Two cases are also possible:
\begin{itemize}
\item When $B$ moves, it becomes an isolated robot. Then, $B$ is the only isolated robot in the configuration. Thus, it is the only one allowed to move. 
Note that since it is the only isolated robot and since there is an even number of robots in the ring, the configuration reached is not symmetric. When $B$ joins the closest $d$.block $D1$, $D1$ becomes the only biggest $d$.block. Note that the configuration reached is not symmetric since there is only one biggest $d$.block and one d.block of size $s-1$. All the robots that are activated (except $A$) have as a target destination $D1$, when $A$ decides to move, it becomes an isolated robots. Thus all the view of robots that took a snapshot before $A$ moves have a correct target even if it is outdated.
\item When $B$ moves, it joins a $d$.block $D1$. $D1$ becomes the only one biggest $d$.block in the configuration. The only robot that is allowed to move is the robot neighboring to $B$ in the initial {\em BlockMirror} configuration (let this robot be $B'$). Its destination is its adjacent empty node towards $D1$. If $A$ moves before $B'$, even if $B'$ took a snapshot before, its target destination is correct since when $A$ moves, $B'$ is allowed to move and its target destination remains $D1$. If $A$ does not move, then once $B$ moves, it joins $D1$ and the configuration contains exactly one biggest $d$.block even after $A$ moves, thus the target of all the robots is correct ($D1$) even if their view is outdated.    
\end{itemize}

Observe that the same holds when the target $d$.block of both robots $A$ and $B$ is the same $d$.block ($D1=D2$) since there will be in the configuration exactly one biggest $d$.block, the target of all the robots is the only biggest $d$.block in the configuration.

We can deduce from the case above that starting from any non-periodic {\em BlockMirror2} configuration without
any tower, there exist at most one outdated robot with an incorrect target and the lemma holds.

\end{proof}

\begin{lemma}\label{blockmirror}
Starting from any non-periodic initial configuration without tower, the configuration cannot become of type {\em BlockMirror} by the behavior of Phase $1$.
\end{lemma}

\begin{proof}
From any non-periodic initial configuration without tower, we assume that the configuration becomes of type {\em BlockMirror} during the execution of Phase $1$.
In the configuration of type {\em BlockMirror}, all the robots belong to $d$.blocks. All the $d$.blocks have the same size $s$ and the number of such $d$.blocks is greater than $2$.
Therefore, in the configuration $C$, before becoming of type {\em BlockMirror}, either there existed one or two $d$.blocks such that their sizes are $s-1$ or there existed one or two isolated robots.

Let us first consider the case where there were no outdated robots with the incorrect target in $C$.
By this algorithm, because the number of robots allowed to move at the same time is at most two, there exists at least one $d$.block $D$ of size $s$. 
That is, the configuration is of type {\em BigBlock}. Two cases are possible as follow:
\begin{itemize}
\item There is an isolated robot neighbor of $D$ in $C$. The closest isolated robots are the only ones that can move. Their destination is their adjacent empty node towards $D$. They keep moving until they reach $D$. Therefore, the $d$.blocks of size $s-1$ cannot become larger by the algorithm of type {\em BigBlock1}.
\item If there is no an isolated robot neighbor of $D$ in $C$ then, the robots that are part of $d$.blocks smaller $D$ which are neighbor of $D$ are the ones allowed to move. Their destination is their adjacent empty node towards $D$. Thus, in this case too, the $d$.blocks of size $s-1$ cannot become larger.
\end{itemize}

Let us now consider the case where there exist outdated robots with incorrect targets in $C$.
By lemma~\ref{outdated}--~\ref{outdated-special}, before the configuration becomes {\em BlockMirror}, the number of outdated robots with incorrect targets is at most one.
If there exists an outdated robot $A$ with the an incorrect target that makes the configuration type {\em BlockMirror}, then there exists at least two $d$.blocks of size $s$, a $d$.block of size $s-1$, and the outdated robot that is an isolated robot.
such that if it moves, then it joins the $d$.block of size $s-1$.
When $A$ becomes an outdated, then its symmetric robot $B$ has moved and has joined a biggest $d$.block.
However, before $B$ moves, there existed a $d$.block of size $s$, such that the destination of $A$ and $B$ are $d$.blocks of the same size $s$.
This is a contradiction.
If there existed exactly one biggest $d$.block after $B$ joins a $d$.block.
By this algorithm, before $A$ moves, the configuration cannot contain two $d$.blocks of size $s$.
This is a contradiction.

Thus, we can deduce that, the configuration cannot become of type {\em BlockMirror} dynamically (It can only be an initial configuration) and the lemma holds.
\end{proof}

\begin{lemma}\label{tower-proof}
Starting from any non-periodic initial configuration, no tower is created by the behavior of Phase $1$.
\end{lemma}

\begin{proof}
If each robot that is allowed to move immediately moves until other robots take new snapshots, that is, no robot has outdated view, then it is clear that no tower is created.

Assume that a tower is created.
Then, there exists a robot $A$ with outdated view in a configuration $C$, and another robot $B$ is allowed to move face-to-face with $A$ in $C$.
When $A$ becomes an outdated robot, the configuration is symmetric and there exists another robot $A^\prime$ allowed to move.
By the proof of lemmas~\ref{outdated} and ~\ref{outdated-special}, if there exist two or more outdated robots, then the configuration is type {\em BigBlock}.
By lemmas~\ref{outdated}-~\ref{outdated-special}, the number of outdated robots with incorrect target is at most one.

Additionally, if $C$ is an initial configuration, then there exists no outdated robot, and otherwise, $C$ is not {\em BlockMirror} by lemma~\ref{blockmirror}.
Therefore, we can consider $C$ is {\em BlockDistance} or {\em BigBlock}.

\begin{itemize}
\item $C$ is type {\em BlockDistance}. By the definition of $C$, there is either one $d$.block of size $k$ or two $d$.blocks of size $k/2$. when $A^\prime$ moves, it creates a new $d'$.block such that $d'=d-1$. Note that since there is in the configuration only one $d'$.block, the new target of all the robots in the configuration is this $d'$.block. 
($i$) if the robot that is at distance $d$ from $A$ becomes allowed to move, its destination is its adjacent empty node in the opposite direction of $A$. Thus, if $A$ decides to move, no tower is created. Contradiction.
($ii$) if $A$ is allowed to move before it neighboring robot at distance $d$ (let this robot be the robot $E$), then we are sure that $A$ moves to an empty node towards $E$ such that a new $d'$.block is created. Thus, in this case too no tower is created. Contradiction.  

\item $C$ is type {\em BigBlock}.
It is clear that in the case the configuration is of type {\em BigBlock1-1}, no tower is created since when $A'$ moves, the only robot that can move is $A$ ($B$ is not allowed to move). For the other cases, since there is, in the configuration, at least one biggest $d$.block, if $B$ is allowed to move,
then, the destination of $B$ is the biggest $d$.block. 
Note that other robots than $A$ can have an outdated view but robot $A$ is the only one having an outdated view with an incorrect target (refer to Lemmas~\ref{outdated}-\ref{outdated-special}).
If there are other robots with an outdated view, since their destination is correct, their target is a biggest $d$.block in the configuration. 
The cases bellow are possible:
\begin{enumerate}
\item \label{C1OTR} If $A^\prime$ joins a $d$.block in the configuration $C$ (let this $d$.block be $D1$), then $D1$ is the only biggest $d$.block in $C$.
Since there is only one biggest $d$.block in the configuration, if a robot is allowed to move, its destination is its adjacent empty node towards the biggest $d$.block (refer to {\em BigBlock} configuration). Thus, the destination of $B$ is its adjacent empty node towards $D1$. Note that there in this case, no other robot between $D1$ and $B$. Therefore $A$ and $B$ cannot move towards each other (face-to-face).  Hence, No tower is created.


\item If $A^\prime$ becomes isolated robot in $C$ and if $A^\prime$ belonged to a $d$.block having a size bigger than $2$, then $A^\prime$ is the only isolated robot that is the closest to a biggest $d$.block.
Then, only $A^\prime$ is allowed to move, and $B$ is not allowed to move. When $A'$ joins the biggest d.block we retrieve Case \ref{C1OTR}.
\item If $A^\prime$ becomes an isolated robot in $C$, and if $A^\prime$ belonged to $d$.block of size equal to $2$, then $A^\prime$ is the only robot that is the closest to a d.block. only $A^\prime$ is allowed to move, and $B$ is not allowed to move. When $A'$ joins the biggest d.block we retrieve Case \ref{C1OTR}.
%
\end{enumerate}

\end{itemize}
\end{proof}

\begin{lemma}\label{periodic-proof}
From any non-periodic initial configuration without any tower, the algorithm does not create a periodic configuration by the behavior of Phase $1$.
\end{lemma}

\begin{proof}
Assume that, after a robot $A$ moves, the system reaches a periodic configuration $C^*$.
Let $C$ be the configuration that $A$ observed to decide the movement.
The important remark is that, since we assume an odd number of nodes, any periodic configuration should have at least three odd number of $d$.blocks with the same size or at least three odd number of isolated robots.
By lemmas~\ref{outdated}-~\ref{outdated-special}, the number of outdated robot with incorrect target is at most one.
By the proof of lemmas~\ref{outdated} and ~\ref{outdated-special}, if there exist two outdated robots, then the configuration is type {\em BigBlock} and not symmetric, and there exists exactly one biggest $d$.block.
Additionally, by the proof of lemmas~\ref{outdated} and ~\ref{tower-proof}, the outdated robot cannot be on the biggest $d$.block and it has a neighboring empty node.

\begin{itemize}
\item $C$ is a configuration of type {\em BlockDistance}.
In $C$, there are only two $d$.blocks of the same size or a single $d$.block such as $d>1$.
Because an outdated robot cannot be on the biggest $d$.block, there exists no outdated robot in $C$.
Then, $C$ is symmetric and there is another robot $B$ that is allowed to move, we can define an odd size hole $H$ on the axis of symmetry, and $A$ and $B$ are neighboring to $H$.
After configuration $C$, three cases are possible: $A$ moves before $B$ moves, $A$ moves after $B$ moves, or $A$ and $B$ move at the same time.
In the case that $A$ moves before $B$ moves, then $A$ and its destination robot become exactly one $(d-1)$.block.
In the case that $A$ and $B$ move at the same time, $A$ and $B$ construct $(d-1)$.block with their destination robots respectively, then there exist exactly two $(d-1)$.blocks.
In the case that $A$ moves after $B$ moves, then $B$ and its destination robot become exactly one $(d-1)$.block.
If other robots move after $B$ moves before $A$ moves, their destination is the $(d-1)$.block, so there exists exactly one $(d-1)$.block.
Therefore, in all cases, $C^*$ is not periodic.

\item $C$ is a configuration of type {\em BlockMirror}.
By lemma~\ref{blockmirror}, $C$ is the initial configuration and there exists no outdated robot.
Then, in $C$, there are only $d$.blocks of the same size and no isolated robot.
\begin{itemize}
\item If $C$ is {\em BlockMirror1}, $A$ is the only one robot allowed to move.
When $A$ moves, it either becomes an isolated robot or joins another $d$.block.
In the latter case, $C^*$ has exactly one biggest $d$.block.
Consider the former case.
If the size of $d$.block $D$ to which $A$ belonged is 2, then $A$ and $D$ becomes isolated robots, otherwise, $A$ becomes exactly one isolated robot.
Thus, in both cases, the number of isolated robots is at most 2.
Therefore, in these cases, $C^*$ is not periodic.
\item If $C$ is {\em BlockMirror2}, there is another robot $B$ that is allowed to move.
After configuration $C$, three cases are possible: $A$ moves before $B$ moves, $A$ moves after $B$ moves, or $A$ and $B$ move at the same time.
In the case that $A$ moves before $B$ moves, then $A$ becomes an isolated robot or joins the guide $d$.block $D1$.
In the previous case, it is clear that $C^*$ is not periodic because the number of isolated robots is at most 2.
In the later case, $D1$ becomes exactly one biggest $d$.block, thus $C^*$ is not periodic.
In the case that $A$ moves after $B$ moves, then there exists exactly one biggest $d$.block to which $B$ joined or two isolated robot $A$ and $B$.
If there exists exactly one biggest $d$.block $D2$ to which $B$ joined, then other robots move to $D2$.
Therefore, there exists exactly one biggest $d$.block before $A$ moves.
If there exist two isolated robots $A$ and $B$, then by {\em BigBlock}, no other robot can move before both of them join $d$.blocks.
In these cases, it is clear that $C^*$ is not periodic.
In the case that $A$ and $B$ move at the same time, $C^*$ is not periodic similarly to the previous case.
\end{itemize}

\item $C$ is a configuration of type {\em BigBlock}.
In this configuration, there can be more than 1 outdated robots, but it is not in the biggest $d$.block.
Additionally, only one of them can have incorrect target.
Let $s$ be the size of biggest $d$.blocks in $C$ and let $R$ be the robot that has an outdated view. Two cases are possible as follow: 
\begin{itemize}
\item $C$ is not symmetric. In this case only $A$ is allowed to move. The sub-cases bellow are possible:
         \begin{itemize}
         \item $A$ and $R$ move and join the same d.block. In this case we are sure that the configuration reached is not symmetric since the configuration contains exactly one biggest d.block.
         \item $A$ and $R$ move and becomes (remain) isolated robots such that the distance between each of them and a biggest d.block is different. In this case too the configuration that is reached is not periodic since if $C*$ is periodic, since $A$ (respectively $R$) is the only closest robot to a d.block otherwise $A$ would not have been allowed to move in $C$.
         \item $A$ and $R$ move and becomes (remain) isolated robots such that the distance between each of them and a biggest d.block is the same. Recall that since the ring has an odd size, if a configuration is periodic then there are at least three odd number of d.blocks. Thus when $A$ and $R$ move, if the configuration that is reached is periodic that means that there is another d.block in $C*$ that has an isolated robot at the same distance as the one between $A$ and its closest d.block. Which is impossible since $A$ was allowed to move in $C$.
         \item $A$ and $R$ move and join different d.blocks. Let refer to these two d.block by respectively $D_1$ and $D_2$. If $|D_1| \ne |D_2|$ then the configuration contains exactly one biggest d.block. Thus, the configuration that is reached is in this case not periodic. If $|D_1|=|D_2|$, then if the configuration that is reached is periodic, since the size of the ring is odd, we are sure that there is another d.block $D_3$ such that $|D_3|=|D_1|$. That is impossible since $A$ was allowed to move in $C$.
         \end{itemize}


\item If $C$ is symmetric. In this case, two robots $A$ and $B$ are allowed to move in $C$. The following cases are possible:
\begin{itemize}
\item $A$ moves before $B$. $C^*$ is not periodic similarly to the non-symmetric case. Note that this holds also when $B$ moves before $A$.
\item $A$ and $B$ move at the same time. Two sub-cases are possible:
          \begin{itemize}
              \item The target of A and B is the same d.block (let refer to this d.block by $D$). In this case, when $A$ and $B$ move, they become the closest ones to $D$. If $R$ moves to a biggest d.block (let this d.block be $D'$) then we are sure that $D \ne D'$ and the configuration reached is not periodic otherwise there were another robot $O$ in $C$ that was closer to $D'$ then $A$ and $B$ to $D$ ($D'$ must have two neighboring robots that are the closest), thus A and B were not allowed to move in $C$.
              \item The target d.block of $A$ and $B$ is different. In this case too we are sure that the configuration reached is not periodic since $C$ was symmetric, there was another robot $O$ that was symmetric to $R$ in $C$. When $A$ and $B$ moves respectively and when $R$ moves according to its outdated view, the configuration reached is not periodic because of robot $O$.
          \end{itemize}

\end{itemize}
\end{itemize}
\end{itemize}
For all cases, $C^*$ is not periodic and the lemma holds.

\end{proof}

 
%
%
%
%

To prove the convergence of the Phase $1$ algorithm, we first show that the size of at least one biggest $d$.block increases. Next, we show that a {\em BlockDistance} configuration with an interdistance $d$ is reached in a finite time. We then prove that the interdistance decreases such that at the end either {\em Terminal} configuration or a configuration $C^{*} \in C_{sp}$ is reached in a finite time. In the latter case Phase $2$ Algorithm is executed.

\begin{lemma}\label{toblockdistance}
Let $C$ be a configuration such that its inter-distance is $d$ and the size of the biggest $d$.block is $s$, and if the configuration can keep symmetric, $s<k/2$, otherwise $s<k$.
From configuration $C$, the configuration becomes one such that the size of the biggest $d$.block is at least $s+1$ in $O(n)$ rounds.
\end{lemma}
\begin{proof}
From configurations of types {\em BlockMirror} and {\em BigBlock}, at least one robot neighboring to the biggest $d$.block is allowed to move.
Consequently, the robot moves in $O(1)$ rounds.
If the robot joins the biggest $d$.block, the lemma holds.

If the robot becomes an isolated robot, the robot is allowed to move toward the biggest $d$.block by configuration of type {\em BigBlock1}.
Consequently the robot joins the biggest $d$.block in $O(n)$ rounds, and thus the lemma holds.
\end{proof}

\begin{lemma}\label{inter-distance}
Let $C$ be a configuration such that its inter-distance is $d$.
From configuration $C$, the configuration becomes type {\em BlockDistance} in $O(kn)$ rounds.
\end{lemma}
\begin{proof}
From lemma~\ref{toblockdistance}, the size of the biggest $d$.block becomes larger in $O(n)$ rounds.
Thus, the size of the biggest $d$.block becomes $k$ or $k/2$ in $O(kn)$ rounds.
Since the configuration that has a single $d$.block with size $k$ or two $d$.block with size $k/2$ is the one, the lemma holds.
\end{proof}

\begin{lemma}\label{minus}
Let $C$ be a configuration which is type {\em BlockDistance} where its inter-distance is $d$($d>1$).
From configuration $C$, the configuration becomes type {\em BlockDistance} where its inter-distance is $d-1$ in $O(kn)$ rounds.
\end{lemma}
\begin{proof}
From the configuration of {\em BlockDistance}, the configuration becomes one such that there is $(d-1)$.block in $O(1)$ rounds.
After that, the configuration becomes one such that there is only single $(d-1)$.block or two $(d-1)$.blocks with same size in $O(kn)$ rounds by lemma~\ref{inter-distance}.
Therefore, the lemma holds.
\end{proof}

\begin{lemma}\label{lem-time}
From any non-periodic initial configuration $C$ without any tower, a Terminal configuration or a configuration $C^{*} \in C_{sp}$ is reached  in $O(n^{2})$ rounds.
\end{lemma}
\begin{proof}
Let $d$ be the inter-distance of the initial configuration.
From the initial configuration, the configuration becomes type {\em BlockDistance} where its inter-distance is $d-1$ in $O(kn)$ rounds by lemma~\ref{inter-distance}.
Since the inter-distance becomes smaller in $O(kn)$ rounds by lemma~\ref{minus}, the configuration becomes one such that there is only single $1$.block or two $1$.blocks of the same size in $O(dkn)$ rounds.
Since $d<n/k+1$ holds, the lemma holds.
\end{proof}
\subsection{Phase $2$}

To prove the correctness of our Phase $2$ algorithm, we first prove that when a configuration $C \in C_{sp}$ is reached, $C$ does not contain any outdated robots with incorrect targets.

\begin{lemma}\label{outdated-phase2}
By the behavior of Phase $1$, the configuration cannot reach a configuration in $C_{sp}$ with outdated robots with incorrect targets.
\end{lemma}

\begin{proof}
To construct a configuration $C^*\in C_{sp}$ with outdated robots with incorrect target during Phase $1$, 
two robots $A$ and $B$ were allowed to move in a symmetric configuration $C$, but the scheduler activated only $B$, and other robots take snapshots after the movement of $B$ before $A$ moves.

\begin{itemize}
\item $C$ is a configuration of type {\em BlockDistance}, where the value of $d$ is $d'>1$.
After $B$ moves, the value of $d$ decreases, there exists exactly one biggest $(d'-1)$.block $D$ and all other robots become isolated robots.
Then, although each robot moves to the biggest $(d'-1)$.block, only the isolated robot which is the closest to the biggest $(d'-1)$.block can move and others cannot move by the configuration of type {\em BigBlock1}.
Because $A$ and all other robots than the biggest $(d'-1)$.block are isolated robots, it cannot create other $(d'-1)$.block before $A$ moves.
Therefore, before $A$ moves, the configuration cannot become a configuration of type {\em BlockDistance} where the value of $d$ is $d'-1$.
So, if the value of $d'$ is more than 2, it is clear that $C^*$ cannot become a configuration in $C_{sp}$ before $A$ moves.
Consider the case that the value of $d'$ is equal to 2.
Because other robot than the biggest 1.block cannot create other 1.blocks which sizes are more than 1, 
there exist isolated robots and only one 1.block.
Therefore, it is clear that the configuration cannot become a configuration in $C_{sp}$ other than {\em Biblock}.
In {\em Biblock}, $A$ seems like an isolated robot.
However, the distance between the biggest 1.block and $A$ is 2.
In $C$, the destination of $A$ is the guide block $D1$ and the $D1$ is neighboring to $H$.
Because other robots than $A$ creates the biggest 1.block by moving towards other guide block $D2$,
the distance between the biggest 1.block and $A$ should be more than $|H|+1$ before $A$ moves.
Therefore, before $A$ move, the configuration cannot become in $C_{sp}$.

\item $C$ is a configuration of type {\em BigBlock1}.
Then, $A$ and $B$ are isolated robots in $C$.
After $B$ moves before $B$ joins the biggest $d$.block, $B$ becomes the nearest isolated robot to the biggest $d$.block.
Therefore, other robots cannot move by the new observation.
\begin{itemize}
\item If $C^*$ becomes one of {\em Start}, {\em Even-T}, {\em Split-A}, {\em Odd-T}, {\em Block}, {\em Biblock}, and {\em TriBlock-A}, 
because there exist two isolated robots, this is a contradiction.
\item If $C^*$ becomes {\em Split-S} or {\em TriBlock-S}, because $A$ and $B$ are isolated robots and $C$ is symmetric before only $B$ moves,
this is a contradiction.
\end{itemize}
After $B$ joins the biggest $d$.block, there exists exactly one biggest $d$.block.
By the configuration of type {\em BigBlock}, the nearest isolated robots or the borders of other $d$.blocks to the biggest $d$.block are allowed to move.
\begin{itemize}
\item If $C^*$ becomes either {\em Start} or {\em Block}, because $A$ is isolated robot, this is a contradiction.
\item If $C^*$ becomes $Event-T$ (resp. {\em Odd-T}), then $B$ belonged to the biggest $d$.block and $A$ is isolated robot.
However, before $B$ belongs to the biggest $d$.block, then there exist two biggest $d$.block and two isolated robots by the definition of {\em Even-T} (resp. {\em Odd-T}).
Then, because $C$ is symmetric and not in $C_{sp}$, this is a contradiction.
\item If $C^*$ becomes {\em Split-S}, because there exist two biggest $d$.block, this is a contradiction.
\item If $C^*$ becomes {\em Split-A}, then either $L1$ or $S2$ is $A$ and either $S1$ or $L2$ is the biggest $d$.block.
If $L1$ is $A$, then the size of $S1$ is $k/2-1$, of $L2$ is 2 and of $S2$ is $k/2-2$ (See the definition of {\em Split-A}.)
Then, the distance between $L1$ and $S1$ is 2 and the distance cannot change before $A$ moves.
Additionally, because the size of $S1$ before $B$ joined is $k/2-2$, $S2$ is also biggest $d$.block in $C$.
However, because only the nearest isolated robot to the biggest $d$.block are allowed to move and other robot cannot move,
the member of $L2$ cannot move.
Because $C$ should be symmetric, this is a contradiction.
\item If $C^*$ becomes either {\em Biblock} or {\em TriBlock-S}, then because $C$ is symmetric and not in $C_{sp}$, this is a contradiction. 
\item If $C^*$ becomes either {\em TriBlock-A}, then $A$ can be the smaller side $d$.block or $B1$ (See the definition of {\em TriBlock-A}.)
If $A$ is $B1$, then $B$ is in the bigger side $d$.block.
However, before $B$ moves, the size of the bigger side and smaller side $d$.block is same and the distance between $A$ and both side $d$.blocks is 2 respectively.
Because $C$ is symmetric, this is a contradiction.
If $A$ is smaller side $d$.block, then $B$ is in $B1$ because $k>8$ and the size of other side $d$.block $D1$ is 2.
However, by the configuration of type {\em BigBlock}, the only neighboring robots to the biggest $d$.block are allowed to move.
Thus, the member robot of $D1$ which is not neighboring to $B1$ cannot move from $C$.
This is a contradiction.
\end{itemize}
 
\item $C$ is a configuration of type {\em BigBlock2}.
Then, $A$ and $B$ are members of $d$.blocks in $C$.
In this case, $B$ becomes isolated robot or joins another $d$.block.
Consider the case that $B$ becomes an isolated robot, then the $d$.block $D1$ to which $B$ belonged becomes an isolated robot or remains in a $d$.block by the movement of $B$.
If $D1$ remains in a $d$.block, then $B$ is exactly one isolated robot, and other robot cannot move by phase 1.
Then, because $C$ is not in $C_{sp}$, $C^*$ cannot be in $C_{sp}$.
If $D1$ becomes isolated, then $A$ is in a $d$.block with size 2 in $C^*$.
Then, by the definition of {\em BigBlock2}, $B$ is the nearest isolated robot to the biggest $d$.block and $D1$ cannot move.
Before $B$ joins another $d$.block, there exist two isolated robots.
Therefore, because $k>8$, the configuration cannot become in $C_{sp}$ before $A$ moves or $B$ joins another $d$.block.

After $B$ joins another $d$.block, then there exists exactly one biggest $d$.block and only the neighbors of the biggest $d$.block can move.
\begin{itemize}
\item If $C^*$ becomes {\em Start} or {\em Split-S}, because there exist more than one biggest $d$.blocks, this is a contradiction.
\item If $C^*$ becomes {\em Even-T} or {\em Odd-T}, there exists an isolated robot.
However, then the isolated robot is a member of non-biggest $d$.block in $C$.
Therefore, because $C$ is not in $C_{sp}$, this is a contradiction.
\item If $C^*$ becomes {\em Split-A}, then $B$ is a border of $S1$ (resp. $L2$).
If $B$ is neighboring to $L1$ (resp. $S2$), then $C$ is {\em Split-S} and this is a contradiction.
If $B$ is neighboring to $S2$ (resp. $L1$) , then $B$ is a member of $S2$ in $C$.
Because $|S1|=|S2|+1$ in $C^*$, the size of $S2$ (resp. $L1$) is bigger than of $S1$ (resp. $L2$) in $C$, this is a contradiction..
\item If $C^*$ becomes {\em Block}, because $A$ does not belong to the biggest $d$.block in $C^*$, this is a contradiction.
\item If $C^*$ becomes {\em Biblock}, then $A$ becomes the isolated robot because $A$ does not belong to the biggest $d$.block.
However, the distance between the biggest $d$.block and $A$ is 2.
This means that, in $C$, the distance between the border and $A$ is 2 and the position of the border does not change.
However, by the behavior of Phase $1$, only the closest robot to the $d$.block are allowed to move.
Therefore, it is a contradiction.
\item If $C^*$ becomes {\em TriBlock-S}, because the distance between the biggest $d$.block and other $d$.blocks is 2 respectively and the size of other $d$.blocks is same in $C^*$, it is a contradiction.
\item Consider the case that $C^*$ becomes {\em TriBlock-A}.
If the size of $B1$ is the biggest, then $B$ is in the smallest $d$.block in $C$.
However, it is {\em TriBlock-S}, this is a contradiction.
If the size of the larger $d$.block other than $B1$ is the biggest and $B$ is in $B1$ in $C$, then it is {\em TriBlock-S}, this is a contradiction.
If the size of the larger $d$.block other than $B1$ is the biggest and $B$ is in the smaller $d$.block other than $B1$ in $C$, then the size of the $d$.block to which $B$ belonged is the biggest, this is a contradiction.
\end{itemize}

\item $C$ is a configuration of type {\em BlockMirror2}.
In $C$, all the robots belong to $d$.blocks, all the $d$.blocks have the same and the number of $d$.blocks is more than 2. If $B$ becomes exactly one isolated robot, then other robot cannot move.
Therefore, $C^*$ cannot become in $C_{sp}$ while $B$ is exactly one isolated robot. If $B$ and the member of $d$.block to which $B$ belonged become isolated robots, then only both of them can move and
the size of $d$.block in $C$ is 2. Therefore, $C^*$ cannot become in $C_{sp}$ while there exist isolated robots because $k>8$ and the other robots cannot move. Consider the case after $B$ joins
another $d$.block $D1$, then $D1$ becomes exactly one biggest $d$.block and there are no isolated robot.
After that, by the configuration of type {\em BigBlock}, only the neighbors of the biggest $d$.block are allowed to move.
\begin{itemize}
\item If $C^*$ becomes either {\em Start} or {\em Split-S}, there exist two biggest $d$.block before $A$ moves, this is a contradiction.
\item If $C^*$ becomes either {\em Even-T} or {\em Odd-T}, $B$ belongs to the biggest block and the isolated robot belonged to the smaller block.
However, then there exist two biggest $d$.block before $A$ moves, this is a contradiction.
\item If $C^*$ becomes {\em Split-A}, $S1$ or $L2$ is the biggest $d$.block in $C^*$. Let $S1$ be the biggest. Because the size of other block than the biggest one and neighbor one to the biggest
block does not change before $A$ moves, the size of $L2$ is the original size. Let $|S1|$, $|S2|$, $|L1|$, and $|L2|$ be the size of each block. Because the size of $S2$ does not increase, we have
$|S2|\le|L2|$. From $|S2|=|S1|-1$ and $|S1|>|L2|$, we have $|S2|=|L2|$. Since $|L1|=|L2|-1$ holds, $B$ is in $L1$ in $C$ and moves to $S1$ in $C^*$. Then $C$ is {\em Split-S}, and this is a contradiction.
We can similarly prove the case that $L2$ is the biggest in $C^*$.
\item If $C^*$ becomes {\em Block}, it is clear that this is a contradiction before $A$ moves.
\item If $C^*$ becomes {\em Biblock}, then $A$ is the isolated robot and other robots move toward the biggest $d$.block.
However, because the distance between $A$ and the biggest $d$.block is 2 and the nearest robots to the biggest $d$.block are allowed to move by {\em BigBlock}, it is a contradiction.
\item Consider the case that $C^*$ becomes {\em TriBlock-S}.
If the $d$.block on the axis of symmetry is the biggest one, then the size of $d$.block to which $A$ belong does not change.
Therefore, it is a contradiction.
Otherwise, there exist two biggest $d$.block before $A$ moves, this is a contradiction.
\item Consider the case that $C^*$ becomes {\em TriBlock-A}. If $A$ and $B$ belong to the same block in $C$, since the distance between the middle block and the side block is two in $C^*$, $C$ is {\em
TriBlock-S}. Thus, $A$ and $B$ belong to different blocks in $C$, and consequently the size of the block $A$ belongs to is not changed in $C^*$. Let $D1$ be the middle block in $C^*$ and $D2$ and $D3$
be the side blocks. In addition, we assume $|D2|=|D3|+1$. Then, $D1$ or $D2$ is the biggest in $C^*$.

Consider the case that $D1$ is the biggest in $C^*$. Note that the size of $D2$ is not equal to that of $D1$ because otherwise $C$ is not {\em BlockMirror2}. Then, $B$ joins $D1$ and $A$ belongs to
$D2$. Since the distance between $D1$ and $D2$ is two, at most one block joins $D1$ between $C$ and $C^*$ (otherwise $A$ moves earlier than the second block joins $D1$). From $|D2|=|D3|+1$, only $B$
moves between $C$ and $C^*$. This implies $C$ is {\em TriBlock-S}, and this is a contradiction.

Consider the case that $D2$ is the biggest in $C^*$. Note that the size of $D3$ is not equal to that of $D1$ because otherwise $C$ is not {\em BlockMirror2}. Then, $A$ belongs to $D3$ and $D1$ is the
smallest. Since $|D2|=|D3|+1$ and the size of $|D3|$ is not changed from $C$, $|D1|=|D3|-1$ holds. This implies $B$ moves from $D1$ to $D2$, however then $C$ is {\em TriBlock-S}. This is a contradiction.
\end{itemize}
\end{itemize}

Therefore, lemma holds.

\end{proof}

\begin{figure}[t]
\begin{center}
\epsfig{figure=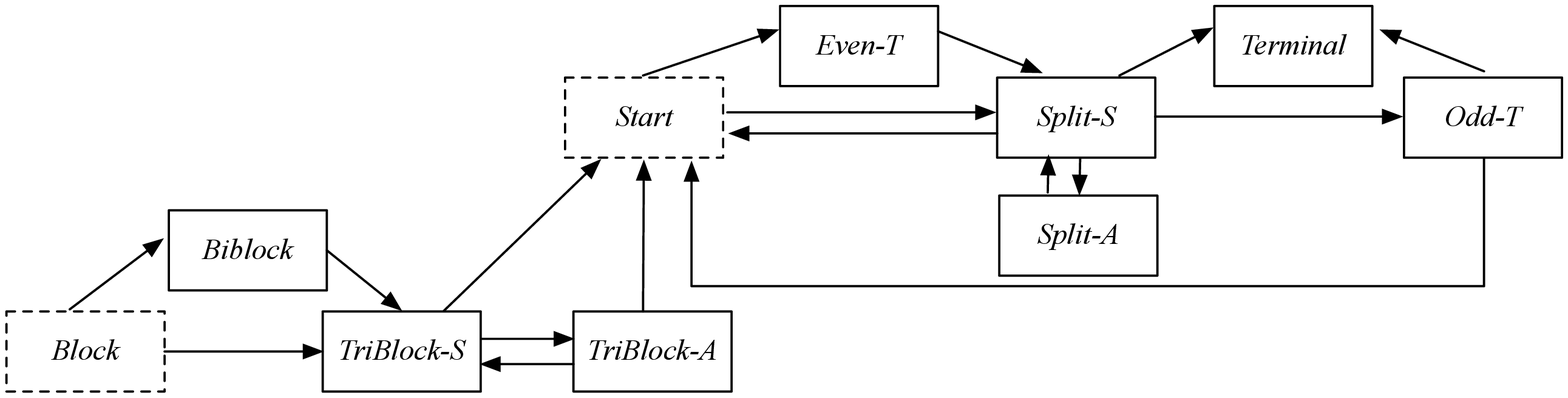,width=11.5cm}
\caption{Transition of the configuration by Phase $2$}
\label{TransPhase2}
\end{center}
\end{figure}


Figure~\ref{TransPhase2} shows the transitions between the $C_{sp}$ configurations when Phase $2$ Algorithm is executed. The symmetry of the configuration is maintained in the following matter: If the scheduler activates only one robot in the configuration then in the next step the robot that was supposed to move is the only one that can move. Note that this robot can be easily determined since we have an odd number of nodes  and an even number of robots in the ring. Thus the robots keep moving in the same direction, towards the {\em Guide Hole} whose size decreases at each time. Hence, {\em Terminal} configuration is reached in a finite time. The algorithm of Phase $3$ can then be executed.

\begin{lemma} \label{EVT}
Starting from a configuration of type {\em Even-T}, the configuration will be of type {\em Split-S} in $O(1)$ rounds.
\end{lemma}

\begin{proof}
The robot allowed to move in this case is the one that is at the border of the 1.block of size $k/2$ sharing a hole with the isolated robot. 
Its destination is its adjacent empty node towards the isolated robot. 
Once this robot moves the configuration will be symmetric and will contain two 1.blocks of size $k/2-1$ and two isolated robots that are at distance $2$ from each 1.block. 
Recall that the isolated robots are considered to be 1.blocks of size $1$. 
Thus the configuration is of type {\em Split-S} and the lemma holds.
\end{proof}

\begin{lemma} \label{STRT}
Starting from a configuration of type {\em Start}, the configuration will be type {\em Split-S}, and the size of the Leader hole decreases by $2$ in $O(1)$ rounds. 
\end{lemma}

\begin{proof}
In a configuration of type {\em Start}, two robots are allowed to move (Recall that the configuration is symmetric). 
These two robots are neighbors to the Leader hole. 
If the scheduler activates both robots at the same time, once they move, the configuration remains symmetric and will contain four 1.blocks (The isolated robots are considered to be single 1.block of size $1$). 
This means that the configuration becomes {\em Split-S}.
Note that since both robots moved, the size of the Leader hole decreases by $2$ and the lemma holds. 
In the case the scheduler activates only one of the two robots allowed to move, the configuration becomes of type {\em Even-T}. 
From Lemma \ref{EVT}, a configuration of type {\em Split-S} is reached in 1 round. 
Then, since the robot allowed to move is at the border of the 1.block of size $k/2$ sharing a hole with the isolated robot. 
Once it moves, the size of the Leader hole decreases by $2$ and the lemma holds.
\end{proof}

\begin{lemma} \label{ODT}
Starting from a configuration of type {\em Odd-T}, the configuration will be either of type {\em Start} or of type Terminal in $O(1)$ rounds.
\end{lemma}

\begin{proof}
The robot allowed to move in the configuration of type {\em Odd-T} is the isolated robot which is at distance $2$ from the 1.block of size $k/2-1$. 
Its destination is its adjacent node towards the closest 1.block of size $k/2-1$. 
Once this robots moves, it joins the 1.block. 
Thus the configuration will be symmetric and contains only two 1.blocks. 
If the size of the Leader hole is equal to $1$ then the configuration reached is of type {\em Terminal} and the Lemma holds. 
If the size of the Leader hole is bigger than $1$ than the configuration reached is of type {\em Start} and the Lemma holds.  
\end{proof}

\begin{lemma} \label{SPLT}
Starting from a configuration of type {\em Split-A}, the configuration will be of type {\em Split-S} in $O(1)$ rounds.
\end{lemma}

\begin{proof}
Note that in the configuration of type {\em Split-A}, the robot allowed to move is at the border of $S1$, its destination is its adjacent empty node in the the opposite direction of the block its belongs to. 
Note that, once this robot moves, it joins the $L1$ and the configuration becomes symmetric. 
Note that the configuration reached contains four 1.blocks such as the size of the hole between the two 1.blocks that are at the same side of the axes of symmetry is equal to $1$. 
Thus the configuration reached is of type {\em Split-S} and the lemma holds.  
\end{proof}

\begin{lemma}\label{SPLTS}
Starting from a configuration of type {\em Split-S}, the Slave blocks will join the Leader blocks, and the configuration reached will be either Terminal or {\em Start} such as the size of the Slave hole increases by two in $O(k)$ rounds.
\end{lemma}

\begin{proof}
The robots allowed to move in a configuration of type {\em Split-S} are the two robots that are in the border of the Salve block, sharing a hole with the Leader block. 
Their destination is their adjacent empty node towards the Leader block. 
Note that in the case the scheduler activates both robots at the same time, the two robots that have moved join the Leader block (since there was only one empty node between the two blocks). 
If the configuration remains of type {\em Split-S}, the system will have the same behavior. 
Note that the number of robots in the Slave block decreases at each time and since the robots join the Leader block at each time they move, the size of this block increases. 
Thus after at most $k/2$ rounds, all the robots will be part of the Leader blocks (Note that once the lasted robots from the Slave block move the size of the Slave hole decreases by $2$). 
Hence the configuration will contain two 1.blocks of the same size. 
If the size of the Leader hole is equal to $1$, then the configuration is of type {\em Terminal} and the lemma holds. 
In the case that the size of the Leader hole is bigger than $1$, then the configuration is of type {\em Start} and the lemma holds as well. 
In the case the scheduler activates only one of the two robots allowed to move, the configuration becomes asymmetric, either of type {\em Split-A}, however from lemma~\ref{SPLT}, a configuration of type {\em Split} is reached in 1 round (Note that the configuration reached is exactly the same as the one reached when both robots have moved). 
Thus at the end, we are sure that the same configurations that are reached when both robots have moved will be reached when the scheduler activates only one robot from the two robots allowed to move. 
Hence either a configuration of type {\em Terminal} is reached or a configuration of type {\em Start} such as the size of the Slave Hole increases by two is reached and the lemma holds.  
\end{proof}

\begin{lemma}\label{BIBLL}
Starting from a configuration of type {\em Biblock}, a configuration of type {\em TriBlock-S} is reached in $O(1)$ rounds.
\end{lemma}

\begin{proof}
In a configuration of type {\em Biblock}, the robot allowed to move is at the border of the 1.block of size $k-1$ not having one robot at distance $2$ as neighbor, once such a robot moves, the configuration becomes symmetric and will contain three 1.blocks.
Then one 1.block of size $k-2$ is at distance $2$ of two 1.blocks of size $1$. 
Thus the configuration reached is of type {\em TriBlock-S} and the lemma holds.
\end{proof}

\begin{lemma}\label{BBBL}
Starting from a configuration of type {\em Block}, a configuration of type {\em TriBlock-S} is reached in $O(1)$ rounds.
\end{lemma}

\begin{proof}
The robots allowed to move in a configuration of type {\em Block} are at the border of the 1.block, and their destination is their adjacent empty node in the opposite direction of the 1.block they belong to. 
Note that in the case the scheduler activates both robots at the same time, then the configuration remains symmetric and contains one 1.block having one isolated robot at distance $2$ at each side. 
Thus the configuration contains three 1.block and the lemma holds. 
In the case the scheduler activates only one of the two robots allowed to move, then the configuration reached will be of type {\em Biblock}. 
However, from Lemma \ref{BIBLL}, in 1 round, a configuration of type {\em TriBlock-S} is reached and the lemma holds. 
\end{proof}

\begin{lemma}\label{TBBA}
Starting from a configuration of type {\em TriBlock-A}, either a configuration of type {\em Start} or {\em TriBlock-S} is reached in $O(1)$ rounds.
\end{lemma}

\begin{proof}
In a configuration of type {\em TriBlock-A}, there are three 1.blocks such as the middle 1.block $B1$ is at distance $2$ from both other 1.blocks.
Note that, since $n>k+3$, there is only one 1.block which has one 1.block at distance $2$ at each side. 
One robot is allowed to move in such a configuration, and this robot is at the border of $B1$ sharing a hole of size $1$ with smallest 1.block (among the two other 1.blocks than $B1$). 
If this robot is part of a 1.block of size bigger than $1$, then once it moves the configuration becomes symmetric and contains three 1.blocks such as there is one 1.block which is at distance $2$ from the two other 1.blocks. 
Thus the configuration reached is of type {\em TriBlock-S} and the lemma holds. 
If the robot allowed to move is part of a 1.block of size $1$, then once it moves the configuration becomes symmetric and contains two 1.blocks. 
Thus the configuration reached is of type {\em Start} and the lemma holds.
\end{proof}

\begin{lemma}\label{TBBS}
Starting from a configuration of type {\em TriBlock-S}, a configuration of type {\em Start} is reached in $O(k)$ rounds.
\end{lemma}

\begin{proof}
The configuration of type {\em TriBlock-S} is symmetric. 
In such a configuration, robots allowed to move are at the border of the middle 1.block (the 1.block that has one 1.block at distance $2$ at each side). 
Note that, if the scheduler activates both robots at the same time, then the configuration remains symmetric. 
If the size of the middle 1.block was equal to $2$, then the configuration reached is of type {\em Start} and the lemma holds. 
Otherwise (the size of the middle 1.block is bigger than $2$), the configuration remains {\em TriBlock-S}. 
However, the size of the middle 1.block decreases at each time. 
Thus, after at most $k-2$ rounds, we are sure that the configuration reached will be of type {\em Start}. 
In the case the scheduler activates only one of the two robots allowed to move, then the configuration becomes of type {\em TriBlock-A}. 
However, from Lemma \ref{TBBA}, a configuration of type {\em Start} or {\em TriBlock-S} is reached in 1 round. 
Note that, in the second case ({\em TriBlock-S}), the configuration is exactly the same as when both robots move at the same time. 
Thus we are sure that a configuration of type {\em Start} is reached in $O(k)$ rounds.  
\end{proof}

\begin{lemma}
Starting from any configuration of type {\em Block}, {\em Biblock}, {\em TriBlock-S} or {\em TriBlock-A}, a configuration of type {\em Start} is reached in $O(k)$ rounds.   
\end{lemma}

\begin{proof}
From the Lemmas \ref{BBBL}, \ref{BIBLL}, \ref{TBBS} and \ref{TBBA} we can deduct the lemma.	 
\end{proof}

\begin{lemma}\label{lem-time2}
Starting from any configuration $C \in C_{sp}$,
 Terminal configuration is reached in $O(kn)$ rounds. 
\end{lemma}

\begin{proof}
From Lemma \ref{STRT}, \ref{EVT},\ref{SPLT},\ref{ODT} and \ref{SPLTS} we can deduct that a configuration of type {\em Terminal} is reached in $O(kn)$ rounds and the lemma holds. 
\end{proof}

\subsection{Phase $3$}
In the followings, the correctness of our Phase $3$ algorithm is proven. We first show that:

\begin{lemma}\label{no-out-phase2}
Starting from a {\em Terminal} configuration, there are no outdated robots with incorrect targets by Phase $1$ and $2$.
\end{lemma}
\begin{proof}
From lemma~\ref{outdated-phase2}, we are sure that after the execution of Phase $1$, we cannot reach a configuration in $C_{sp}$ with outdated robots that have incorrect target.
Thus, there is no outdated robots with incorrect target when {\em Terminal} configuration is built after the execution of Phase $1$.
Additionally, starting from any configuration in $C_{sp}$, when {\em Terminal} configuration is created, no robot has an incorrect outdated target since, when the configuration is symmetric, if the scheduler activates only one robot, then in the second step, the only robot that is allowed to move is the one that was supposed to move when the configuration was symmetric, so it is like the scheduler activates both robots that are allowed to move at the same time. 
Therefore, we can deduce that in {\em Terminal} configuration, there are no outdated robots with incorrect targets.
\end{proof}

When {\em Terminal} configuration is reached then when Algorithm of Phase $3$ is executed we have the following lemma:
\begin{lemma}\label{target}
Starting from a {\em Terminal} configuration, {\em Target} configuration is reached in $O(1)$ rounds. 
\end{lemma}
\begin{proof}
When the configuration is {\em Terminal}, there are exactly two robots that are allowed to move. 
These two robots are the ones that are at the extremity of the $1$.block having a hole of size $1$ as a neighbor. 
Let refer to these two robots by $r_1$ and $r_2$ such that they are respectively part of the $1$.blocks $B1$ and $B2$. 
Two cases are possible as follow:
\begin{enumerate}
\item \label{P3C1}The scheduler activates both robots at the same time. In this case a tower is created on the axes of symmetry and the configuration becomes a {\em Target} configuration.
\item The scheduler activates only one robot $r_1$. Once $r_1$ moves, it joins $B_2$. Note that $r_1$ and $r_2$ are neighbors. The only robot allowed to move is $r_2$, its destination is the node occupied by $r_1$. Once it moves, a tower is created and we retrieve Case \ref{P3C1}.
\end{enumerate}
From the case above, we can deduce that {\em Target} configuration is reached in $O(1)$ rounds and the lemma holds.
\end{proof}

\begin{lemma}\label{single}
Starting from a {\em Terminal} configuration, a configuration with a single $1$.block of size $k-1$ is created in $O(k)$ rounds.
\end{lemma}
\begin{proof}
By lemma~\ref{target}, the configuration becomes symmetric {\em Target} configuration in which there are two $1$.blocks and one tower. By executing the gathering algorithm in \cite{KameiLOT11}, the robots that are neighbors at distance $2$ from the tower move first, a $1.$block is created on the axes of symmetry, the robots that are at distance $2$ from such a $d$.block are the only ones allowed to move and so on.
By repeating such an execution, a configuration with a single $1$.block of size $k-1$ is reached in $O(k)$ rounds and the lemma holds.
\end{proof}

Finally, we show that the gathering is eventually performed. 

\begin{lemma}\label{gathering-Terminal}
Starting from a {\em Terminal} configuration, the gathering is performed in $O(k^{2})$ rounds.
\end{lemma}
\begin{proof}
By lemma~\ref{single}, a configuration with a single $1$.block of size $k-1$ is reached in a finite time. Observe that the configuration is symmetric and contains a tower on the axes of symmetry. On another hand, the robots not part of the tower can see an odd number of robots in the ring.
By the same proofs in \cite{KameiLOT11}, from a configuration that contains a single $1$.block of an odd size (the robots part of a tower are considered as a single robot), the gathering is performed in $O(k^2)$ rounds. 
\end{proof}

Thus, the following Theorem holds.

\begin{theorem}
Starting from any non-periodic initial configuration without any tower, the gathering is performed in $O(n^{2})$ rounds.
\end{theorem}

\begin{proof}
From Lemmas \ref{lem-time} and \ref{lem-time2}, we can deduce that starting from any non-periodic initial configuration without any tower, a configuration of type {\em Terminal} is reached in $O(n^{2})$ rounds. From Lemma \ref{gathering-Terminal}, we are sure that starting from a {\em Terminal} configuration, the gathering is performed in $O(k^2)$ rounds. Thus the Theorem holds.
\end{proof}

\section{Conclusion}\label{sec:conclusion}

We presented a gathering protocol for an even number of anonymous and oblivious robots that are initially located on different nodes of a ring, and are endowed with a weak local multiplicity detector only. Our gathering can start from any configuration that is not periodic, yet expects the ring to have an odd size. This constraint permits to avoid edge-edge symmetries in the initial configurations, as they are known to be ungatherable~\cite{Klasing08-j}. 

If we relax the constraint on the parity of the ring size (that is, the ring is even) but maintain the absence of edge-edge symmetry and periodicity requirement (that are mandatory for problem solvability), no node-edge symmetry can actually occur in the initial configuration. If the initial configuration is not symmetric, it is known that gathering with local weak multiplicity detection is feasible~\cite{Izumi10}. There remains the case of the initial node-node symmetry. With \emph{global} weak multiplicity detection, this case is solvable for $6$ robots~\cite{ASN11c} or more than $18$ robots~\cite{Klasing08-j}. A similar characterization using only \emph{local} weak multiplicity detection looks challenging.

\bibliographystyle{plain}
\bibliography{biblio} 

\end{document}